\documentclass[review,sort&compress,3p,times]{elsarticle}

\usepackage{ecrc}

\volume{00}

\firstpage{1}

\journalname{Information Sciences}

\runauth{}

\jid{}

\jnltitlelogo{Information Sciences}

\usepackage{amssymb}
\usepackage{amsthm}
\usepackage{url}

\usepackage[figuresright]{rotating}

\newtheorem{mydef}{Definition}
\newtheorem{mytheorem}{Theorem}
\newtheorem{mylemma}{Lemma}

\newtheorem{mycol}{Corollary}
\theoremstyle{remark}
\newtheorem{myremark}{Remark}

\begin{document}

\begin{frontmatter}

\dochead{}

\title{Shannon Entropy based Randomness Measurement and Test \\for Image Encryption}


\author[tufts]{Yue Wu\corref{cor*}}
\ead{ywu03@ece.tufts.edu}
\cortext[cor*]{Corresponding author. Tel: +1 617-627-3217}
\author[tufts]{Joseph P. Noonan}
\author[Texas]{Sos Agaian}

\address[tufts]{Department of Electrical and Computer Engineering, Tufts University, 161 College Ave. Medford, MA 02155, USA}
\address[Texas]{Department of Electrical and Computer Engineering, University of Texas at San Antonio, One UTSA Circle, San Antonio, TX 78249, USA}



\begin{abstract}
The quality of image encryption is commonly measured by the Shannon entropy over the ciphertext image. However, this measurement does not consider to the randomness of local image blocks and is inappropriate for scrambling based image encryption methods. In this paper, a new information entropy-based randomness measurement for image encryption is introduced which, for the first time, answers the question of whether a given ciphertext image is sufficiently random-like. It measures the randomness over the ciphertext in a fairer way by calculating the averaged entropy of a series of small image blocks within the entire test image. In order to fulfill both quantitative and qualitative measurement, the expectation and the variance of this averaged block entropy for a true-random image are strictly derived and corresponding numerical reference tables are also provided. Moreover, a hypothesis test at significance $\alpha$-level is given to help accept or reject the hypothesis that the test image is ideally encrypted/random-like. Simulation results show that the proposed test is able to give both effectively quantitative and qualitative results for image encryption. The same idea can also be applied to measure other digital data, like audio and video.
\end{abstract}


\end{frontmatter}


\section{Introduction}
Recently, the Wikileaks incident harmed many personal and governmental interests by publishing a series of private, secret and classified media. This incident again reminded people and governments the importance of information security, although information security is still a relative new word existing for a half century.
About two decades ago, the first data encryption standard (DES), a block cipher based on a 56-bit key, \cite{DES} was published for binary sequences. Since then, research on the sequence cipher published many excellent ciphers, including Blowfish \cite{Blowfish} in 1993, Twofish \cite{Twofish} in 1998 and advanced encryption standard (AES) \cite{AES} in 1998. Nowadays, these encryption algorithms still prevail and are used by thousands of people, organizations and companies.

As the other side of the coin, cryptography analysis developed at the same time as data encryption.Many ciphers are considered to be insecure due to some undesired properties that are weak to some cryptography analysis. For example, DES is believed to be insecure as it is weak to the differential attack and the bruteforce attack \cite{CryptographyBook}. Many of these attacks, such as the frequency attack, ciphertext-only attack and known ciphertext attack are designed directly for weak ciphertext, which is not random-like. Therefore, the ability of generating random-like ciphertext is one of most crucial criteria for a secure cipher.

A main focus of testing the randomness of a ciphertext is its distribution. Ideally, this distribution of ciphertext is uniform, because a uniform distribution implies that the each symbol in ciphertext is equally important. As a result, the cipher is invulnerable to statistical attacks. Moreover, both the relationship between ciphertext and plaintext and the relationship between ciphertext and encryption key are very complicated and involved. These two properties, namely confusion and diffusion, were identified by Claude Shannon in his 1949 masterpiece paper \cite{ShannonCipher}.

Conventionally, randomness tests are designed for binary sequences, for example the Kolmogorov test \cite{HandbookStatistics}, poker test \cite{HandbookStatistics}, gap test \cite{HandbookStatistics}, autocorrelation test \cite{HandbookStatistics}, Shannon entropy \cite{ShannonEntropy}, diffusion randomness test \cite{RandomnessBlockCipher}, etc. Moreover, the standard randomness tests FIPS 140-1\cite{FIPS140-1} and 140-2\cite{FIPS140-2} are also for binary sequences. However, these randomness tests are out of date for image encryption, because digital image is typically a two dimensional data rather than a one-dimensional binary sequence. More specifically, it is different from a one-dimensional binary sequence in two main aspects: the high information redundancy between neighbor pixels and the bulk data nature. On one hand, this implies that conventional binary sequence ciphers are not good for image data either for a low encryption efficiency \cite{DateEncryption} or a relatively small block size for encryption. Many image ciphers have been researced, including chaotic system based image ciphers \cite{2Dpiecewise,3DBaker,3DCat,ChaoticFinitePrecision,ChaoticNeuralNetwork,Compound3DBaker,CompoundChaos,DNAChaos,HyperChaos,LogisticMap,MitureChaos,SPcipherKumar,ZhangChaos}, SCAN language based algorithms \cite{SCAN,SCAN2}, transform based algorithms \cite{MagicCube,WaveletEncryption,Sudoku} and others \cite{WaveEncryption}. On the other hand, the distinctive characteristics of image data implies that the randomness tests designed for binary sequences are inappropriate for image data.
As a result, randomness tests for these image ciphers are commonly done with respect to some specific attack(s) instead of testing the randomness of binary sequences. For example, histogram analysis \cite{WaveEncryption,Sudoku,2Dpiecewise,3DBaker,3DCat,ChaoticFinitePrecision,ChaoticNeuralNetwork,Compound3DBaker,CompoundChaos,DNAChaos,HyperChaos,LogisticMap,MitureChaos,SPcipherKumar,ZhangChaos,abcd}, and Shannon entropy test  \cite{DNAChaos,BitPermutation,Compound3DBaker,ChaoticNeuralNetwork,MitureChaos,2Dpiecewise}
are designed to test the cipher security to statistical attacks; autocorrelation test \cite{ChaoticNeuralNetwork,DNAChaos,3DCat,BitPermutation,WaveEncryption,MitureChaos,2Dpiecewise,abcd} is designed to test the cipher security to ciphertext, and the Unified Average Changing Intensity (UACI) and Number of Pixel Change Rate (NPCR) tests \cite{DNAChaos,Compound3DBaker,CompoundChaos,BitPermutation,WaveEncryption,3DCat,MitureChaos,2Dpiecewise,abcd}are designed to test the cipher security to differential attacks.

One major problem of these prevailing randomness tests for image ciphers is that these tests provide quantitative results rather than qualitative results. For example, in the Shannon entropy test, the entropy scores for ciphertext images $A$ and $B$ may be 7.9911 and 7.9912, respectively. Although ciphertext image $B$ has a relatively higher entropy score than $A$, which implies that ciphertext $B$ is more random-like than $A$, the most important problem is still unanswered: whether or not the ciphertext image(s) $A$, $B$ or both are sufficiently random-like? In other words, whether a ciphertext is sufficiently random-like is the first priority question. 
As long as a ciphertext is sufficiently random-like, whether its entropy score is 7.9911, or 7.9912 makes no statistical significance of difference.

\begin{enumerate}
    \item This measurement should consider the nature of image data, for example, the two-dimensional property and bulk data property.
    \item This measurement should be mathematically well defined and easily calculated.
    \item This measurement should be applicable to all kind of image encryption methods, including permutation cipher, substitution cipher, etc.
    \item This measurement should provide reference values from a true-random image, such that both the quantitative and qualitative results can be drawn for any given test image.
\end{enumerate}
It is clear that both conventional histogram analysis and Shannon entropy test lack properties 1), 3) and 4) and that autocorrelation analysis, UACI and NPCR attain the property 1) 2) and 3) while not 4). Some other image encryption quality/randomness measurements \cite{Li:2008:NMI:1487744.1488442,edgeRandomness}
 also fail to include the property 4.

In this paper, we proposed a Shannon entropy based randomness test for image encryption which attains all above four properties. Unlike other variants of Shannon entropy for image encryption \cite{Li:2008:NMI:1487744.1488442}, our measurement is directly designed for image data and is well mathematically grounded. The proposed measurement uses information from local image blocks rather than the global image. More specifically, it defines the sample mean $\overline {H_K}$ of information entropies for $K$ randomly selected non-overlapped image blocks within a test image. Both the theoretical mean $\mu_{H_K^*}$ and variance $\sigma_{H_K^*}^2$ of this sample mean are derived from a true-random image, where the intensity of each pixel follows an independent and identical uniform distribution. Quantitative test results are obtained by comparing the actual $\overline {H_K}$ and the theoretical mean $\mu_{H_K^*}$. Moreover, an $\alpha$ level of statistical hypothesis test is also given in the paper to give a qualitative test according to the Central Limit Theorem (CLT). The case of approximating $\overline {H_K}$ as a Gaussian when $K$ is insufficiently large is also considered by using the Berry-Esseen Theorem (BET). Randomness tests by using the block entropy test on commercial image ciphers \cite{I-Cipher,pictureEncryption} and image encryption algorithms \cite{3DCat,Sudoku} are done with computer simulations. Simulation results validate the effectiveness of the test.

The remainder of the paper is organized as follows: Section 2 discusses the preliminary backgrounds in the Shannon entropy, the CLT and the BET; Section 3 proposes the block entropy test for image encryption, derives the theoretical mean $\mu_{H_K^*}$ and variance $\sigma_{H_K^*}^2$ from a true random image and develops a hypothesis test with $\alpha$-level significance; Section 4 provides simulation results of testing randomness on image ciphers by using the block entropy test; and Section 5 concludes the paper.
\section{Background}
\label{Background}
\subsection{Shannon Entropy and Properties}
In 1948, Claude Shannon first proposed the concept of Shannon entropy \cite{ShannonEntropy}. Since then, the Shannon entropy has been widely used in a variety of information sciences. It measures the randomness and quantifies the expected value of the information contained in a message, usually in bit units. The Shannon entropy of a random variable $X$ can be defined as Eqn. (\ref{eqnEntropy}), $P_i$ is defined in Eqn. (\ref{eqnpi}) where $x_i$ indicates the $i$th possible value of $X$ out of $n$, and $P_i$ denotes the possibility of $X=x_i$.
\begin{equation}
\label{eqnEntropy}
H(X) =H(P_1,...,P_n)= -\sum\limits^{n}_{i=1}P_i\log_2 P_i
\end{equation}
\begin{equation}
\label{eqnpi}
P_i = Pr(X=x_i)
\end{equation}
The Shannon Entropy attains, but is not limited to, the following properties:
\begin{enumerate}
    \item Bounded: $0\leq H(X)\leq \log_2{n}$
    \item Symmetry: $H(P_1,P_2,...) = H(P_2,P_1,...)$ etc.
    \item Grouping: $H(P_1,...,P_n) = H(P_1+P_2,P_3,...,P_n)+(P_1+P_2)H(\frac{P_1}{P_1+P_2},\frac{P_2}{P_1+P_2})$
\end{enumerate}

In the context of digital images, an image $X$ of size $M$-by-$N$ can be considered as a system with $L$ pixel intensity scales. For example, a 8-bit gray image allows $L = 256$ gray scales from 0 to 255. Additionally, denote the number of pixels within image $X$ at pixel intensity scale $l$ as $\aleph(l)$. Then $P_l=Pr(X=l)=\aleph(l)/MN$. Finally, Eqn. (\ref{eqnEntropy}) turns into Eqn. (\ref{eqnImgEntropy}) in the context of image data.
\begin{equation}
\label{eqnImgEntropy}
H(X) =-\sum\limits^{L-1}_{l=1}P_l\log_2 P_l=\sum\limits^{L-1}_{l=0}\frac{\aleph(l)}{MN}\log_2\frac{MN}{\aleph(l)}
\end{equation}
This image entropy attains its maximum when a pixel's intensity is equally likely at any scale $l$ as Eqn. (\ref{dUniform}) shows.
\begin{equation}
\label{dUniform}
P_0 = P_1 = \ldots = P_l= \ldots = P_{L-2} = P_{L-1} = 1/L
\end{equation}

An ideally encrypted image is completely random and thus its entropy reaches the theoretical maximum $\log_2L$. Since the image entropy is a quantitative measurement for $\{P_0,...,P_l,...,P_{L-1}\}$, it is an equivalent test to the histogram analysis, which plots the distribution of $P_l$ and is commonly used for security analysis in the image encryption literature.
\subsection{Central Limit Theorem and Berry-Esseen Theorem}
Let $Y_1,Y_2,...,Y_n$ be a sequence of $n$ independent and identically distributed observations on a random variable $Y$ associated with a finite mean $\mu$ and a variance $\sigma^2$. The Central Limit Theorem (CLT) states that the sample mean of these observations approaches the normal distribution with a mean $\mu$ and a variance $\sigma^2/n$, as the number of samples $n$ increases. Mathematically, the CLT can be stated as Eqn. (\ref{CLT}) shows. The most important merit of the CLT is that it is true even when the probability density function (PDF) of the random variable $Y$ does not follow a normal distribution.
\begin{equation}
\label{CLT}
\overline{Y_n}=\sum\limits^{n}_{i=1}\frac{Y_i}{n} \sim {\cal N}(\mu,\frac{\sigma^2}{n}),as {\ \ n\rightarrow\infty}
\end{equation}

Instead of using $\overline{Y_n}$, the random variable $Z_n$ defined in Eqn. (\ref{Z_test}) is commonly used in hypothesis tests, where $Z_n \sim {\cal N}(0,1)$ as $n\rightarrow\infty$. The convergence in distribution implies that Eqn. (\ref{Z_norm}) is held for arbitrary $z\in\Re$, where $\Phi(z)$ is the cumulative distribution function (CDF) of ${\cal N}(0,1)$. Consequently, the statistical test designed on this $Z$ statistic is called the Z-test, where $F_{Z_n}(z)$ denotes the actual CDF of $Z_n$. .
\begin{equation}
\label{Z_test}
Z_n=\frac{\overline{Y_n}-\mu}{\sigma/\sqrt{n}}
\end{equation}
\begin{equation}
\label{Z_norm}
\lim_{n\rightarrow\infty}Pr(Z_n\leq z) = \lim_{n\rightarrow\infty}F_{Z_n}(z)=\Phi(z)
\end{equation}

As a supplement to CLT, Berry-Esseen Theorem (BET) quantifies the rate at which this convergence to normality takes place. The theorem states that there exists a positive constant $\cal C$, such that the BET inequality defined in Eqn. (\ref{BET}) holds for all $z$ and $n$, where $\sigma$ is the standard deviation of $Y$ and $\rho$ is defined in Eqn. (\ref{rho}).
\begin{equation}
\label{BET}
|F_{Z_n}(z)-\Phi(z)|\leq \frac{{\cal C}\rho}{\sigma^3\sqrt n}
\end{equation}
\begin{equation}
\label{rho}
\rho=E[|Y-\mu|^3]<\infty
\end{equation}

The BET tells that the difference between the CDF of the standard normal distribution and that of the $n$ sample mean of a random variable associated with a finite mean $\mu$ and a variance $\sigma^2$ differs by no more than the specified amount. The calculated values of the constant $\cal C$ was 7.59 \cite{BET1942} in 1941. Over the years, this value has greatly decreased. So far the best estimate of $\cal C$ is found as ${\cal C}<0.4784$ \cite{BET} in 2011, derived from the inequality Eqn. (\ref{Korolev}).

\begin{equation}
\label{Korolev}
\sup\limits^{}_{z} |F_{Z_n}(z)-\Phi(z)|\leq\frac{0.33477(\rho+0.429\sigma^3)}{\sigma^3\sqrt{n}}
\end{equation}



\section{Information Entropy based Randomness Test for Image Encryption}
\subsection{Ideally Encrypted Image}
Before designing a measurement or test for image encryption, the question of what is an ideally encrypted image has to be answered. Although there might be other answers, the ideally encrypted image in this paper is considered as a random-like image, which is not discernible from a true random image. In other words, we believe that the ideally encrypted image and the true random image have the same statistics. Indeed, this is the ideal case for image encryption, which encrypts the original image pixel information so well that one cannot differentiate the encrypted image from a true random image. However, sharing the same statistics does not imply that an ideally encrypted image is a true random image but some image very alike it.

\begin{mydef}
True Random Image (Ideally Encrypted Image):\\ If an image random field $R$ of size $M$-by-$N$ with $L$ intensity scales satisfies the condition that $\forall i\in[1,M]$ and $j\in[1,N]$, the image pixel located at the $i$th row and $j$th column $R(i,j)$ is an independently and identically distributed random variable with the discrete uniform distribution over 0 to L-1, namely $\forall \textit{ pixel } p \in R$, $p\sim{\cal U}[0,L-1]$, then this image $R$ is a random image of size $M$-by-$N$ with $L$ intensity scales.
\end{mydef}

It is noticeable that a true random image $R$ has a very flat histogram, a high Shannon entropy over the entire image and a very low autocorrelation coefficient between neighbor pixels. Moreover, any image block of a true random image $R$ should also attain a high Shannon entropy. This is an assumption that is true but is omitted in randomness tests for image encryption. In other words, we believe that an image containing some image blocks with low Shannon entropy scores is not ideally encrypted/random-like, no matter how high its global Shannon entropy is.
%

Fig. 1 shows an example why the additional local randomness constraints has to be included for testing randomness in image encryption. The plaintext image of the rabbit is encrypted by some image cipher using block processing. Some image block, say ears of the rabbit, however, is encrypted by a weak key. From the color histogram, it is clear that such a weak does not influence the ciphertext distributions very much. Although such a weak key only leads to a limited amount of information leakage (only ears are recongized), it may completely divulge the information of the image (rabbit).

\begin{figure}[ht!]
    \label{fig-rabbit}
    \begin{minipage}[b]{0.24\linewidth}
      \centering
     \centerline{\includegraphics[width=\linewidth]{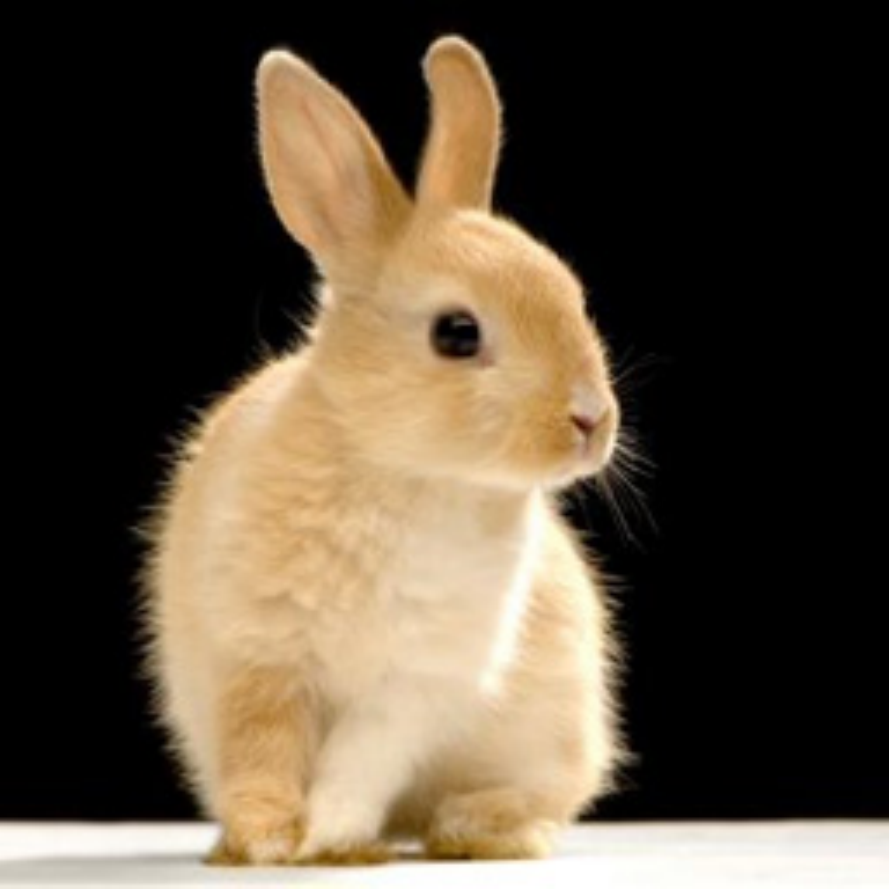}}
      \vspace{0.12cm}
    \end{minipage}\hfill
    \begin{minipage}[b]{0.24\linewidth}
      \centering
     \centerline{\includegraphics[width=\linewidth]{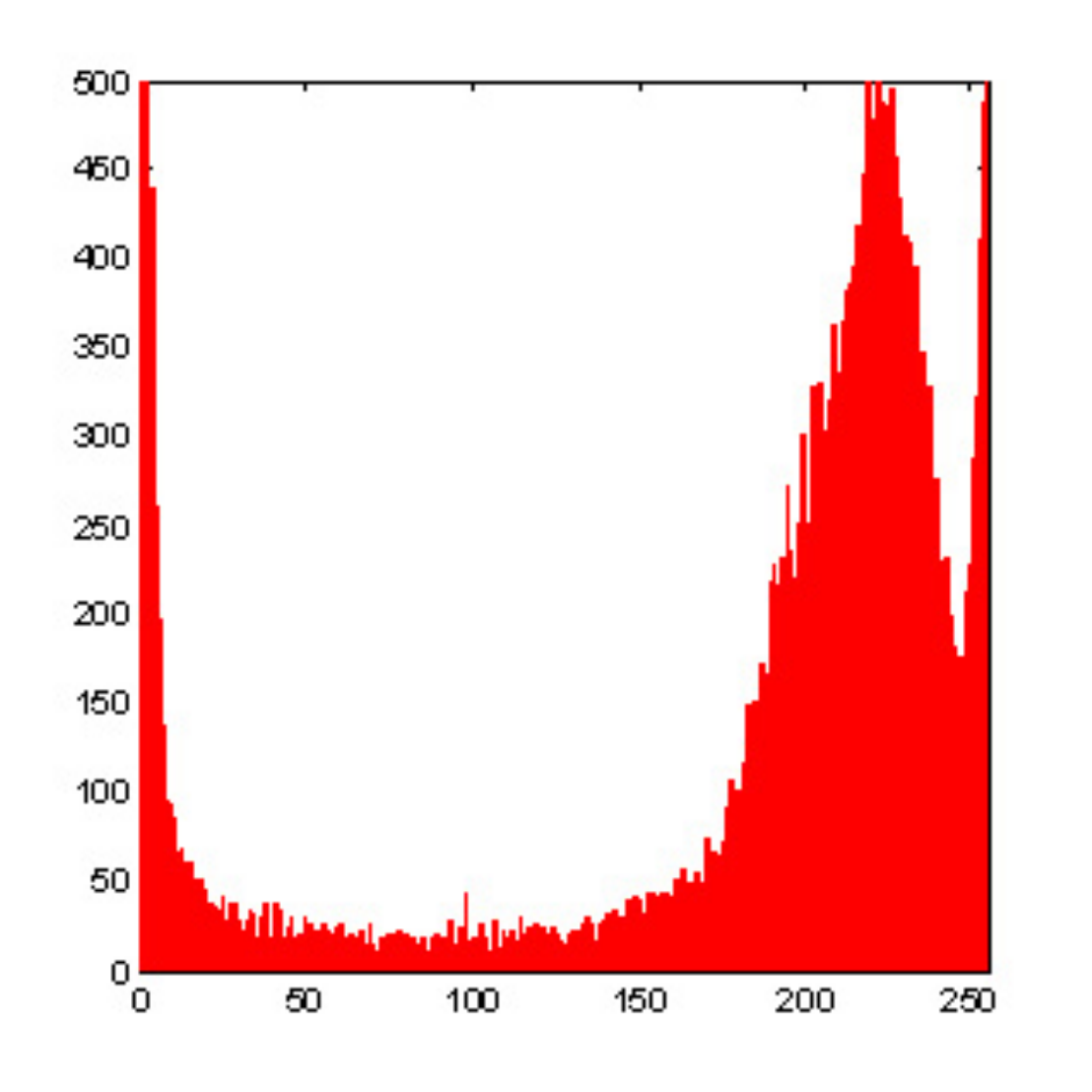}}
      \vspace{0.12cm}
    \end{minipage}\hfill
    \begin{minipage}[b]{0.24\linewidth}
      \centering
     \centerline{\includegraphics[width=\linewidth]{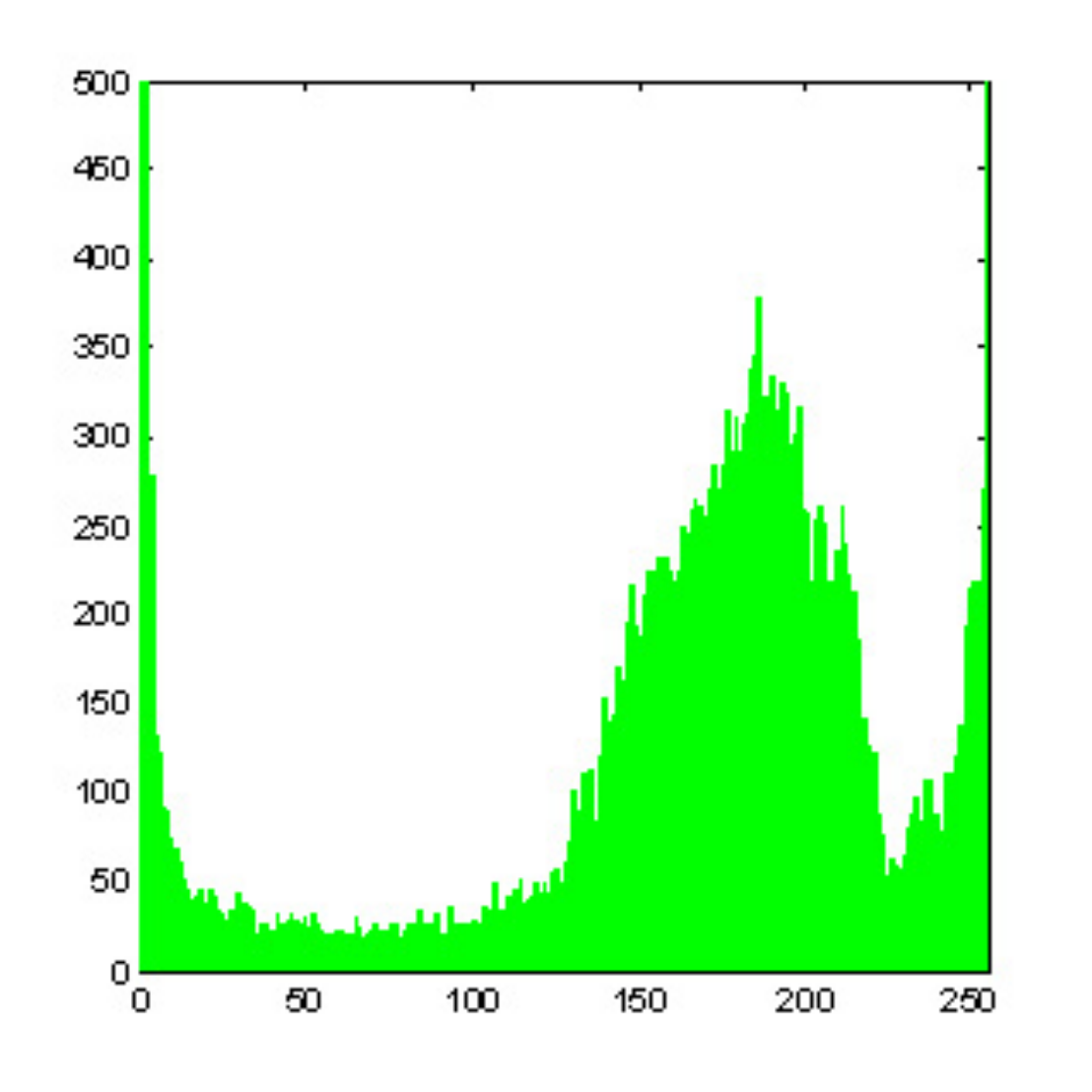}}
      \vspace{0.12cm}
    \end{minipage}\hfill
    \begin{minipage}[b]{0.24\linewidth}
      \centering
     \centerline{\includegraphics[width=\linewidth]{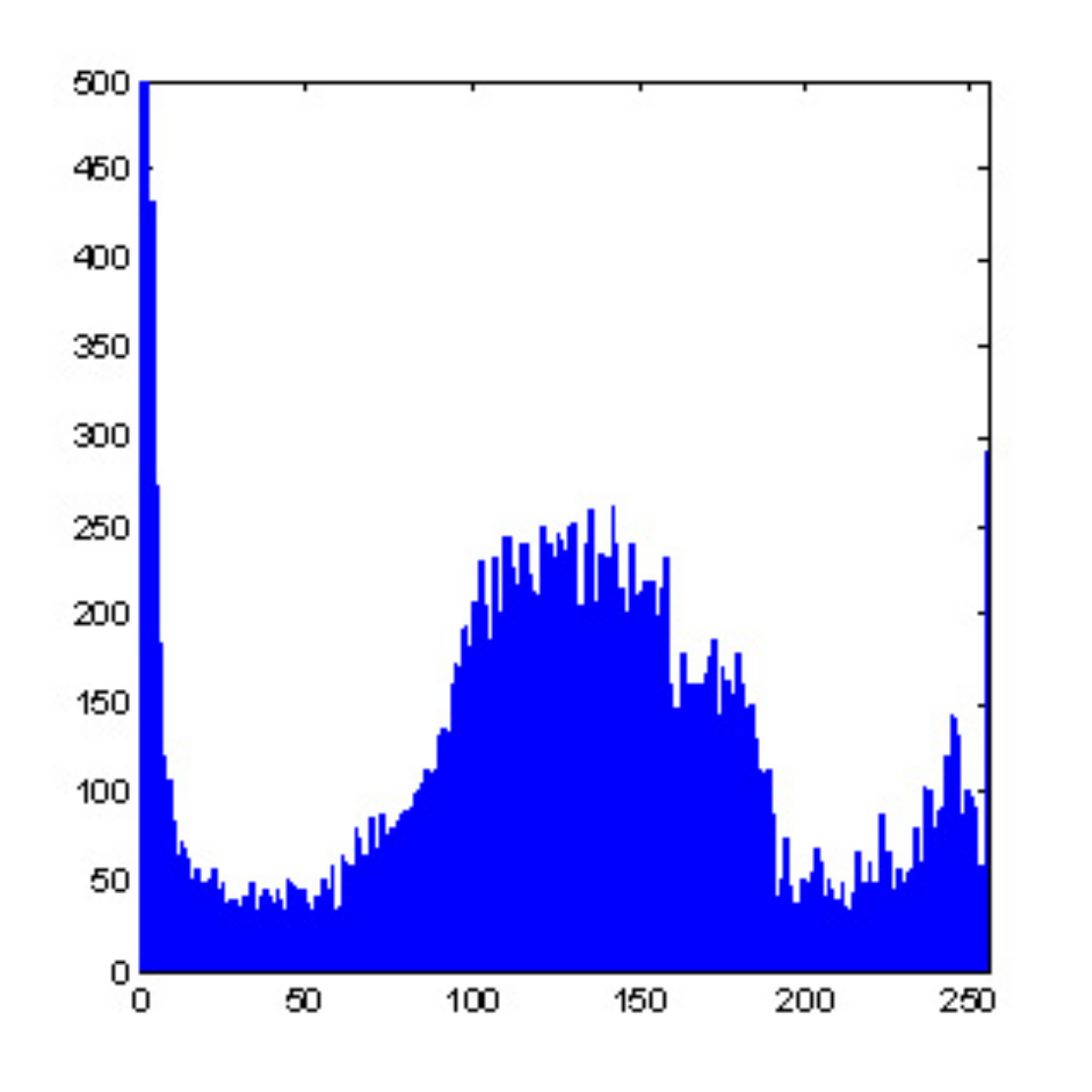}}
      \vspace{0.12cm}
    \end{minipage}\hfill
    \begin{minipage}[b]{0.24\linewidth}
      \centering
     \centerline{\includegraphics[width=\linewidth]{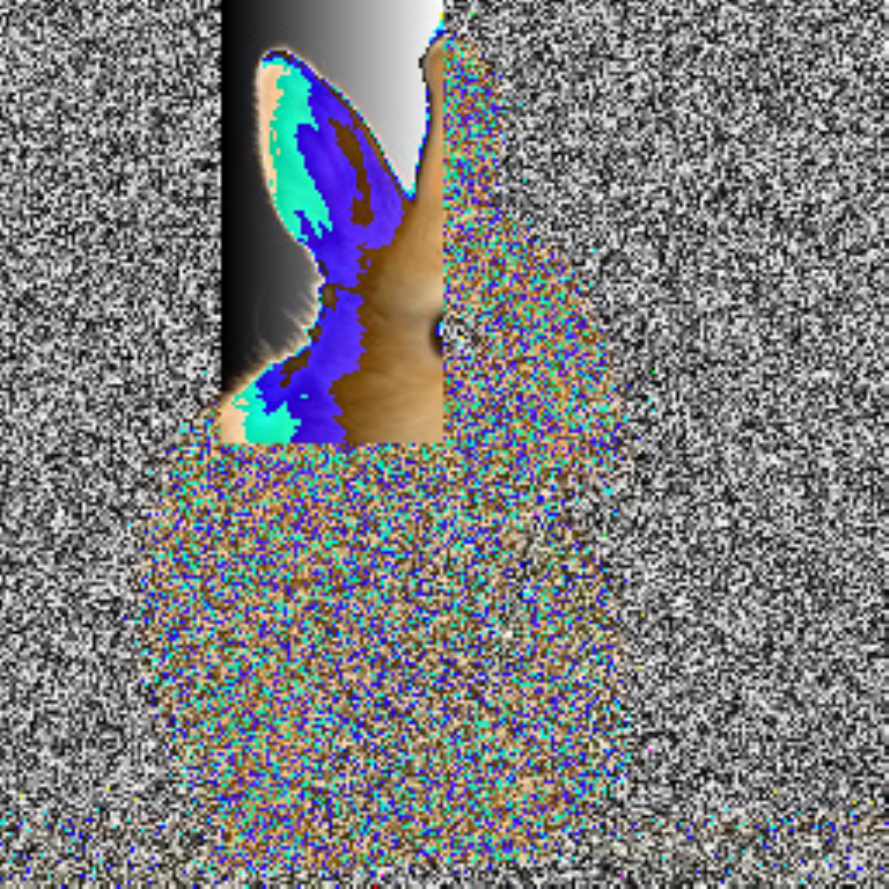}}
      \vspace{0.12cm}
    \end{minipage}\hfill
    \begin{minipage}[b]{.24\linewidth}
      \centering
     \centerline{\includegraphics[width=\linewidth]{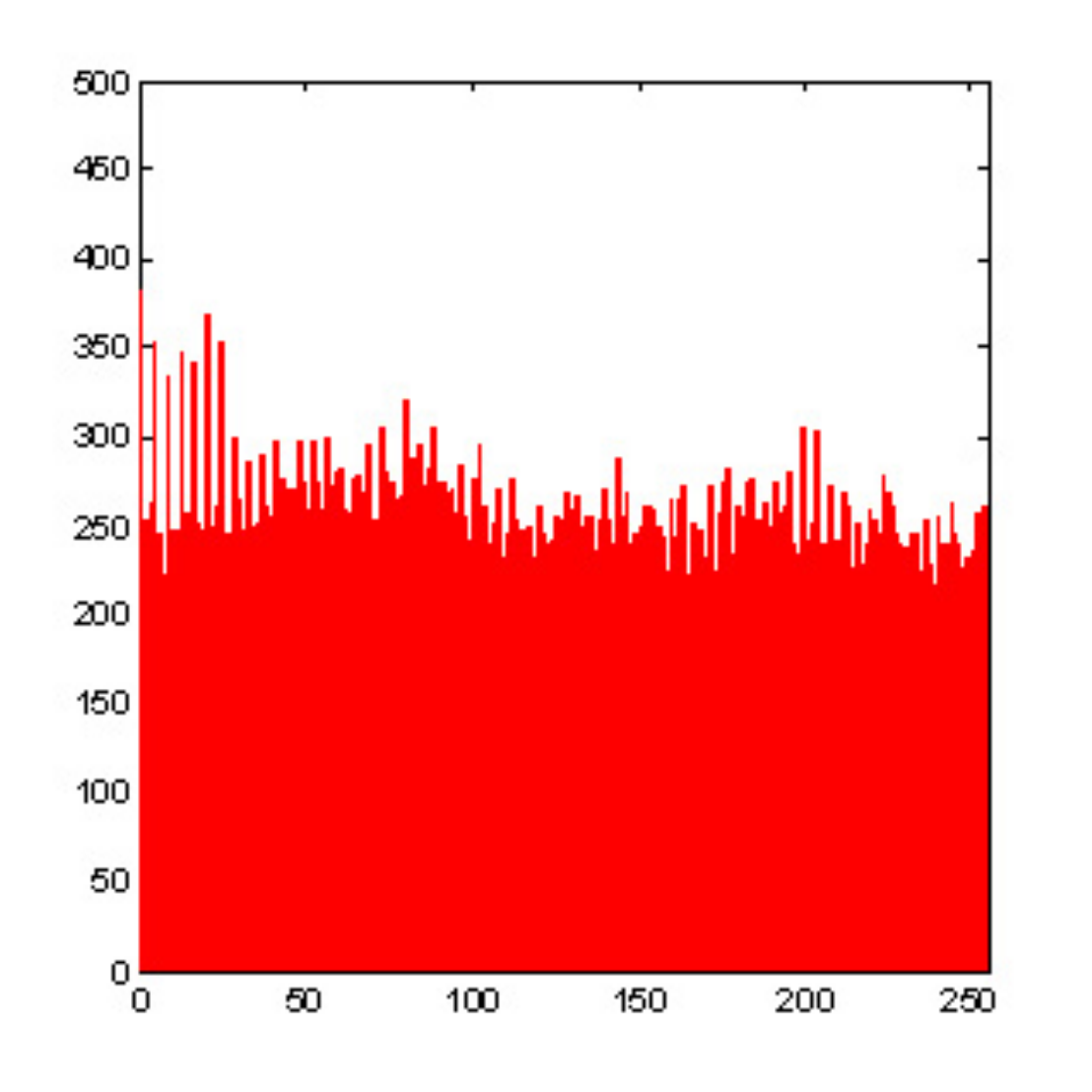}}
      \vspace{0.12cm}
    \end{minipage}\hfill
    \begin{minipage}[b]{.24\linewidth}
      \centering
     \centerline{\includegraphics[width=\linewidth]{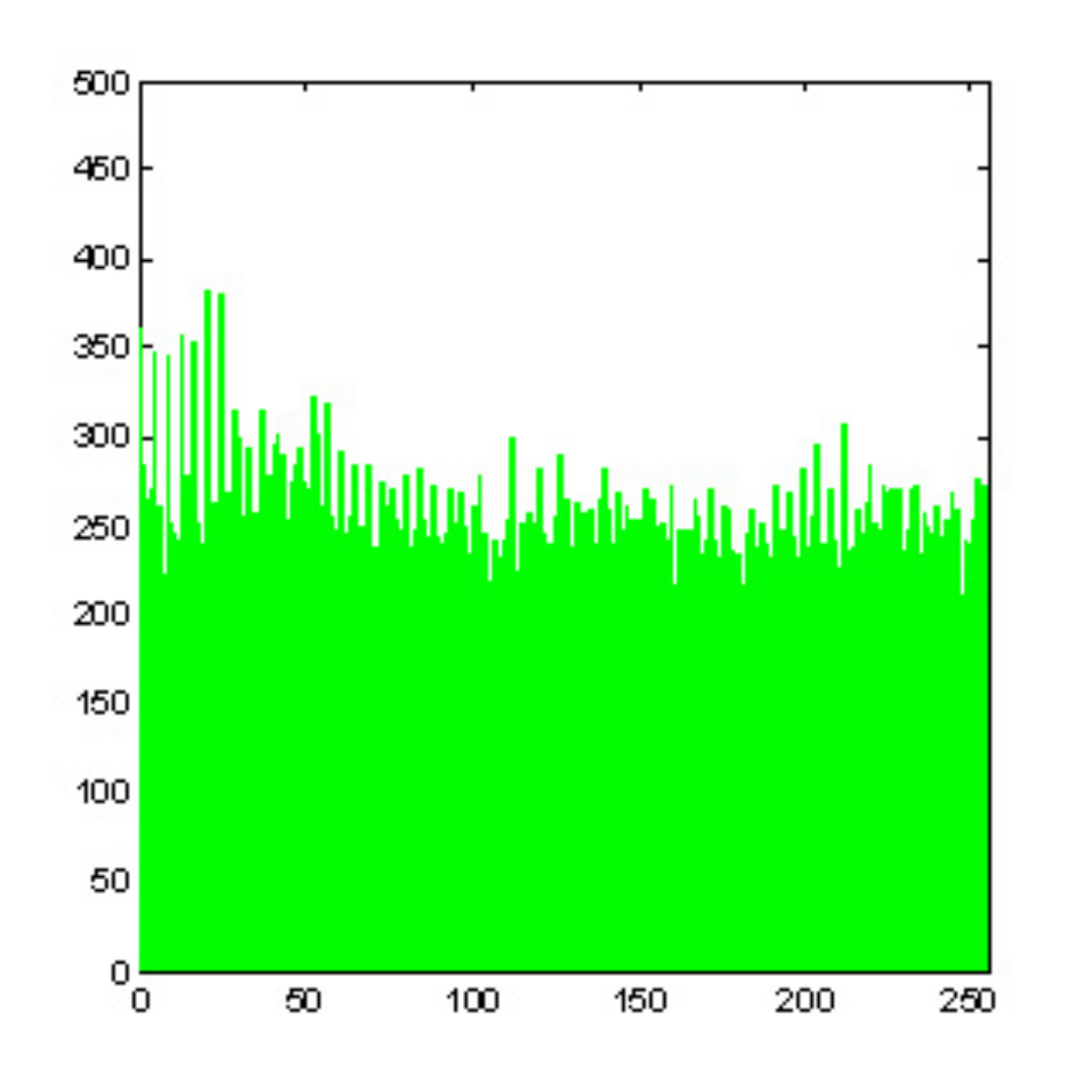}}
      \vspace{0.12cm}
    \end{minipage}\hfill
    \begin{minipage}[b]{.24\linewidth}
      \centering
     \centerline{\includegraphics[width=\linewidth]{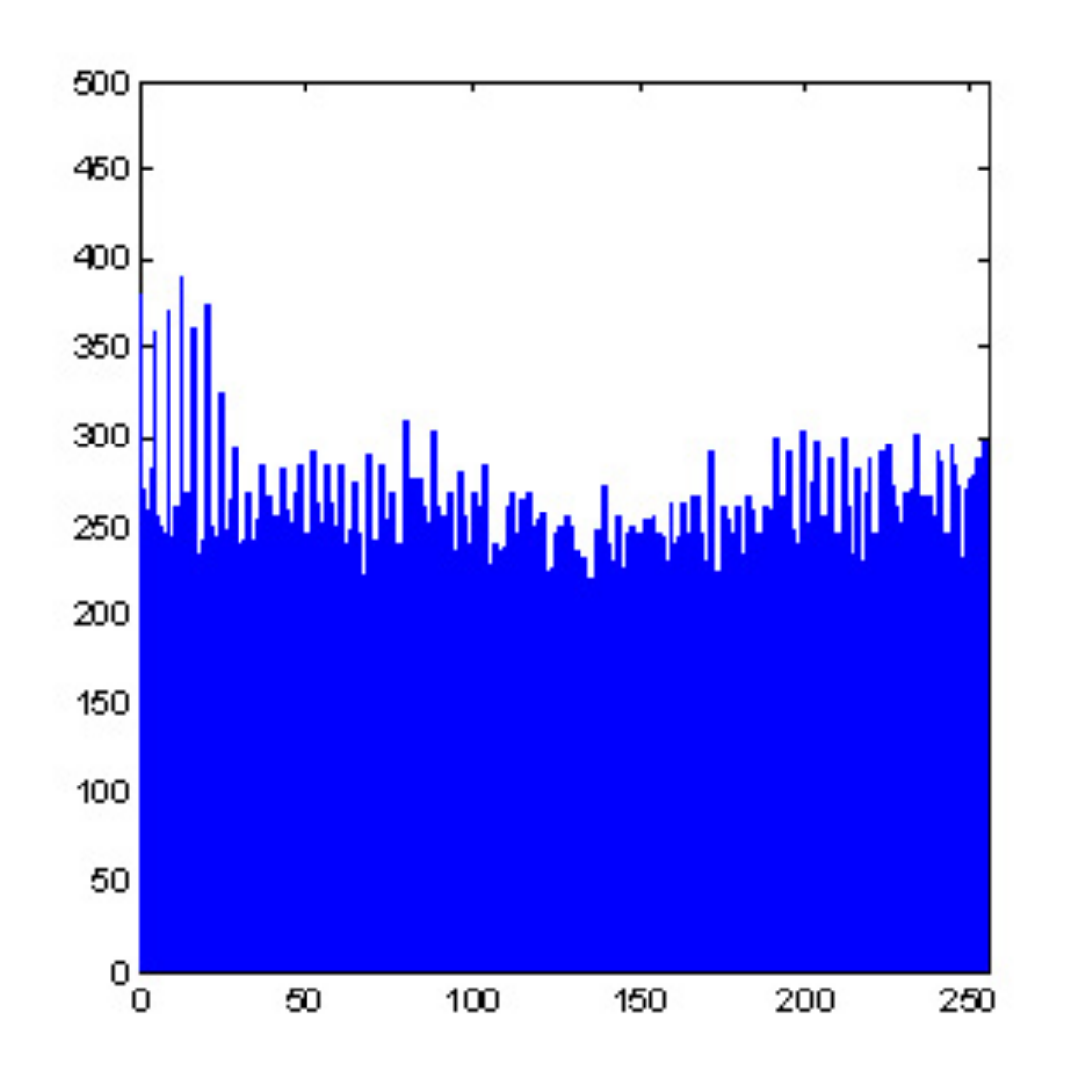}}
      \vspace{0.12cm}
    \end{minipage}\hfill
    \caption{An image encryption example (1st row: plaintext image and its histograms for R,G,and B channels; 2nd row: ciphertext image and its histograms for R,G and B channels)}
\end{figure}

Although the ciphertext histogram distributions contains some small 'bumps', the entropy of ciphertext is 7.99373, which is very close to 8, the theoretical upper bound of entropy for an 8-bit image. Therefore, a pure high entropy tested on the entire image is insufficient to test whether the test image is ideally encrypted/random-like. In other words, the global entropy test does not satisfy the first measurement consideration, laid out in the introduction.
\subsection{Block Entropy Test for Image Encryption}

In order to measure the local entropy over image blocks rather than the global entropy, the block entropy test on a given encrypted image $X$ is fulfilled as follows:
\begin{itemize}
    \item {Step 1. Randomly select $Y_1,Y_2,...,Y_K$ image blocks at size $M$-by-$N$ within $X$ ($L$ intensity scales) without overlapping}
    \item {Step 2. $\forall i\in\{1,2,...,K\}$ calculate Shannon entropy $H(Y_i)$ via Eqn. (\ref{eqnImgEntropy}) }
    \item {Step 3. Calculate the sample mean $\overline{H_K}$ over these $K$ block entropies via Eqn. (\ref{sampleMean})}
\end{itemize}
\begin{equation}
\label{sampleMean}
\overline{H_K}=\sum\limits^{K}_{i=1}\frac{H(Y_i)}{K}
\end{equation}
\begin{figure}[ht!]
    \label{fig-example}
    \begin{minipage}[b]{0.24\linewidth}
      \centering
     \centerline{\includegraphics[width=\linewidth]{rabbit}}
     \centerline{\scriptsize(a) Input}
      \vspace{0.12cm}
    \end{minipage}\hfill
    \begin{minipage}[b]{0.24\linewidth}
      \centering
     \centerline{\includegraphics[width=\linewidth]{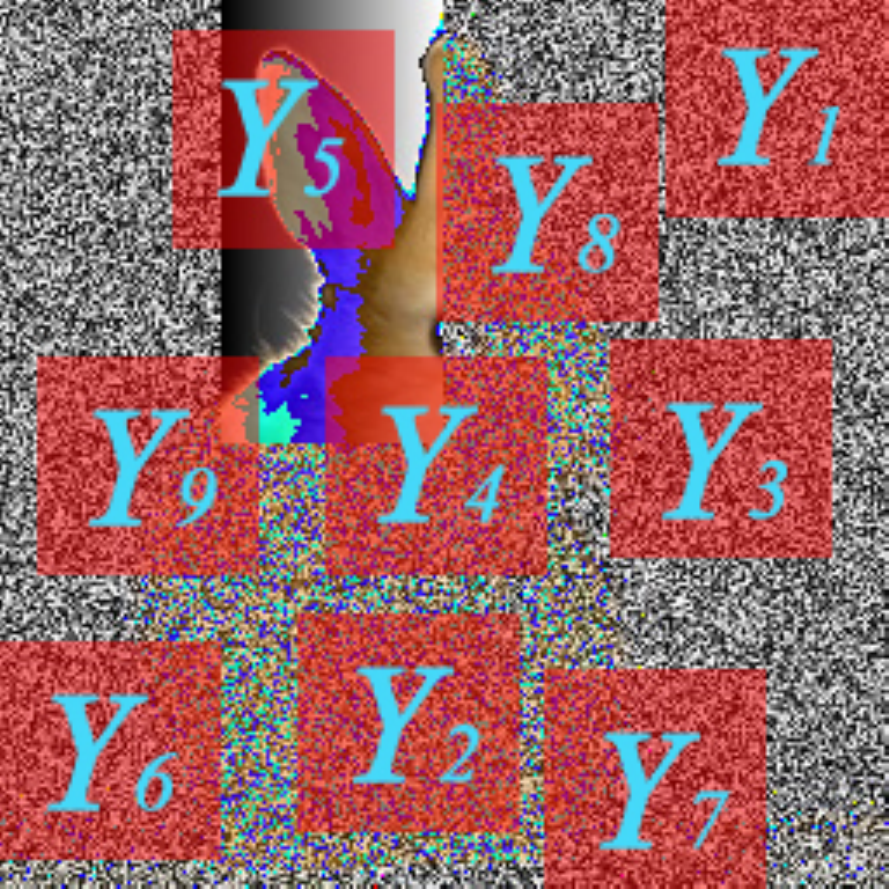}}
     \centerline{\scriptsize(b) Step 1}
      \vspace{0.12cm}
    \end{minipage}\hfill
    \begin{minipage}[b]{0.24\linewidth}
      \centering
     \centerline{\includegraphics[width=\linewidth]{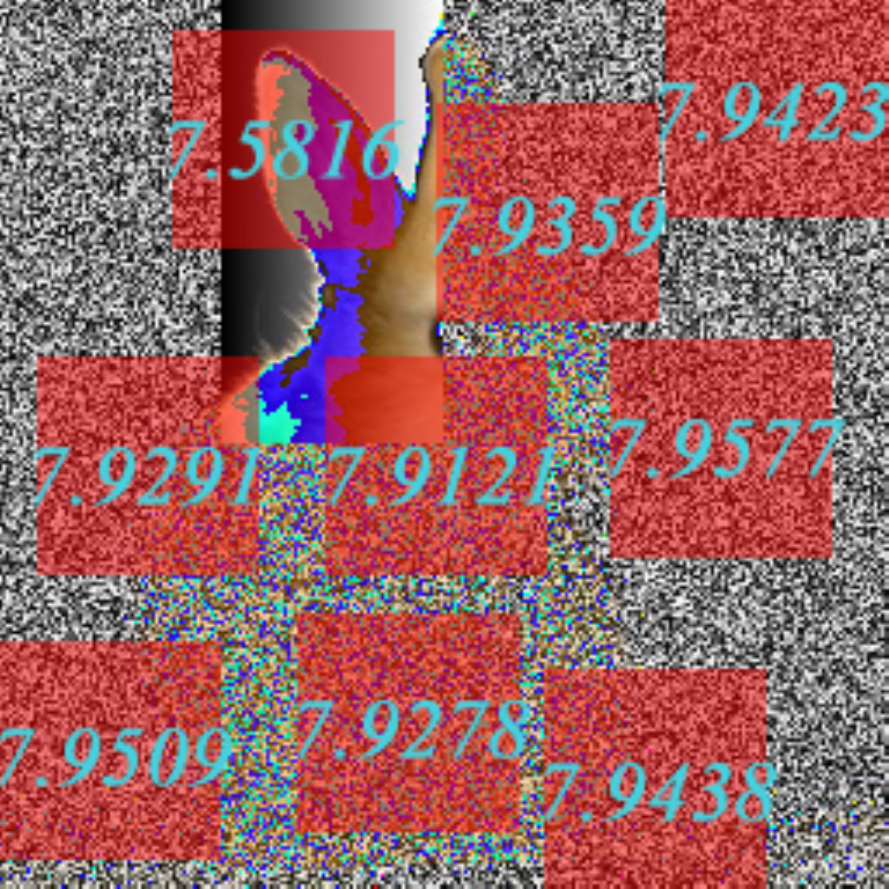}}
     \centerline{\scriptsize(c) Step 2}
      \vspace{0.12cm}
    \end{minipage}\hfill
    \begin{minipage}[b]{0.243\linewidth}
      \centering
     \centerline{\includegraphics[width=\linewidth]{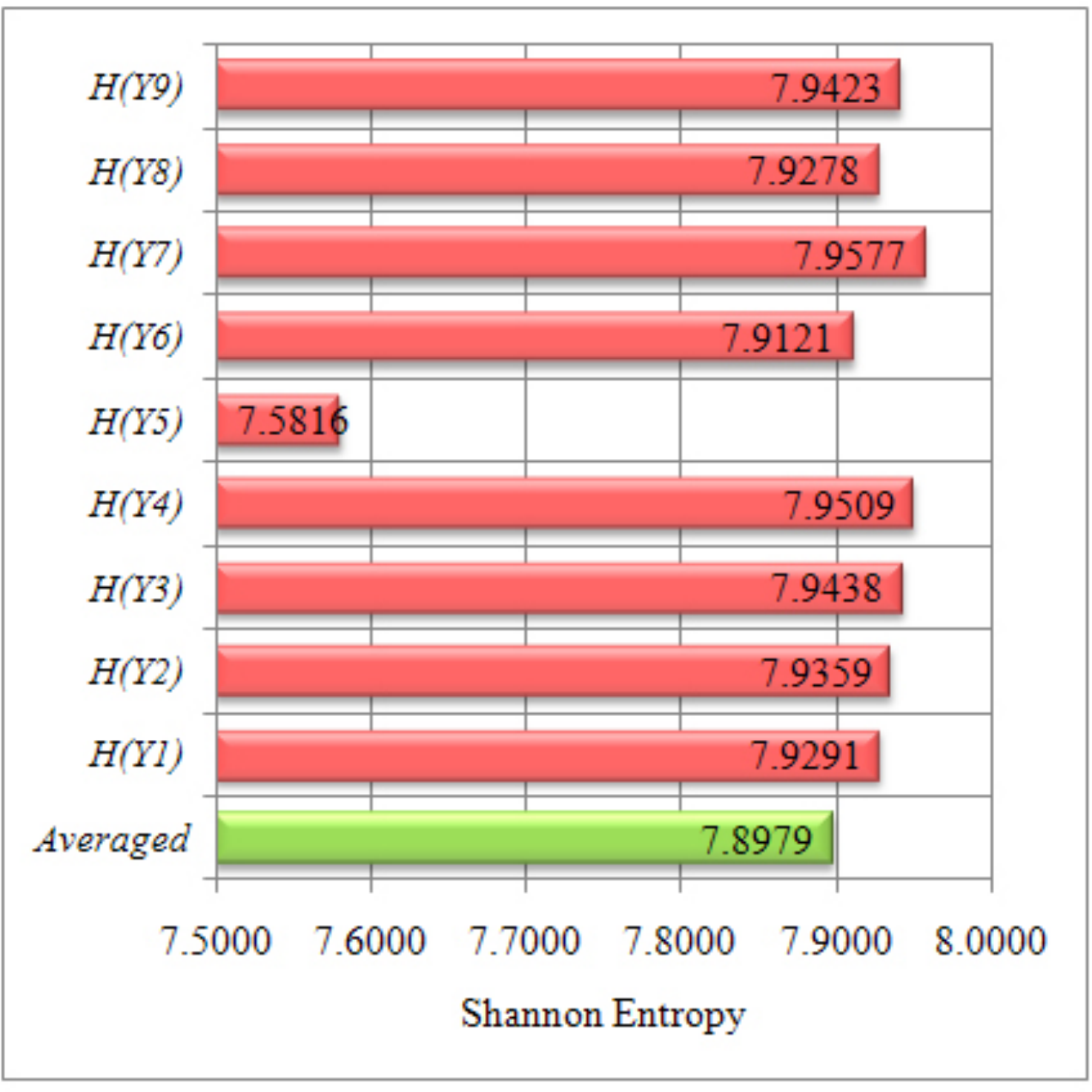}}
     \centerline{\scriptsize(d) Step 3}
      \vspace{0.12cm}
    \end{minipage}\hfill
    \caption{An example of the block entropy test}
\end{figure}

Fig. 2 shows the three steps of block entropy test on the previous encrypted rabbit image. It is clear that the new block entropy test on $\overline{H_K}$ deserves considerations for testing image randomness, because:
\begin{itemize}
    \item Randomly selected and non-overlapped $K$ regions over the entire image give a fair coverage on a two-dimensional image;
    \item The compact definition of $\overline{H_K}$ ensures a light computational cost;
    \item $\overline{H_K}$ is applicable to both permutation and substitution based encryption methods.
\end{itemize}

In order to make $\overline{H_K}$ a meaningful statistic that can be used to judge the randomness in an image, $\overline{H^*_K}$ from the ideally encrypted image has to be derived. In the remainder of this section, the theoretical mean and variance of $\overline{H^*_K}$ are given as a function of the block size $MN$ and of the number of allowed intensity scales $L$.
\subsubsection{Theoretical Mean and Variance of $\overline{H^*_K}$}
In this section, $Y$ is considered as an image block at size $M$-by-$N$ with $L$ intensity scales from an ideally encrypted image. Then $\forall i\in\{1,2,...,K\}$ $Y_i$ is an observation of $Y$. Since none of two tested image blocks are overlapped, it is reasonable to assume they are independent as well. As a result, $H(Y_1),H(Y_2),...,H(Y_K)$ are independently and identically distributed about $f_{H(Y)}$, which is the PDF of $H(Y)$. Hence, as long as $K$ is sufficiently large, the theoretical mean and variance of $\overline{H^*_K}$ is linked to $H(Y)$ via Eqn. (\ref{CLT}). Therefore, the first task is to find the mean and the variance of $H(Y)$.

\begin{mylemma}
     Let $\aleph(l)$ be the number of pixels in $Y$ at intensity scale $l$, then $\aleph(l)$ follows a binomial distribution associated with $MN$ trails with $1/L$ success possibility, i.e.$\aleph(l)\sim {\cal B}i(MN,1/L)$.
\end{mylemma}
\begin{proof}
    \begin{tabbing}
        \hspace{36pt} \=$\because Pr(p=l |p\in Y)=1/L$\\
        \>$\therefore Pr(p\neq l |p\in Y)=(L-1)/L$\\
        \>$\therefore $ The random variable for a pixel $p$ at intensity level $l$ follows a Bernoulli distribution with success\\
        \>\hspace{10pt}possibility $1/L$.\\
        \>$\therefore $ For a $M$-by-$N$ image block $Y$, $\aleph(l)\sim {\cal B}i(MN,1/L)$, that is:
    \end{tabbing}
    \begin{center}
        $Pr(\aleph(l)=n)=\frac{MN!}{n!(MN-n)!}\frac{(L-1)^{MN-n}}{L^{MN}}$
    \end{center}
\end{proof}
\begin{mycol}
\label{colP(l)}
    $Pr(P_l=n/MN)=\frac{MN!}{n!(MN-n)!}\frac{(L-1)^{MN-n}}{L^{MN}}$%
\end{mycol}

\begin{proof}
    \begin{tabbing}
        \hspace{36pt}\=$\because P_l=Pr(Y=l)=\aleph(l)/MN$ and $\aleph(l)\sim {\cal B}i(MN,1/L)$\\
        \>$\therefore Pr(\aleph(l)=n)=Pr(MN\cdot P_l=n)=Pr(P_l=n/MN)=\frac{MN!}{n!(MN-n)!}\frac{(L-1)^{MN-n}}{L^{MN}}$
    \end{tabbing}
\end{proof}
\begin{mytheorem}
\label{TheMultinomial}
    If $\sum\limits^{L-1}_{l=0}n_l = MN$, then $Pr(\aleph(0)=n_0$, $\aleph(1)=n_1, ...$, $\aleph(L-1)=n_{L-1})= {MN!}/({n_0!...n_{L-1}!}{L^{MN}})$
\end{mytheorem}
\begin{proof}
    \begin{tabbing}
        \hspace{36pt}\= Denote event of $\aleph(i)=j$ as $E_i^j$ .\\
        \>$\because$ Lemma 1 holds for an arbitrary scale $l\in\{0,1,...,L-1\}$\\
        \>$\therefore$ $Pr(\aleph(0)=n_0$, \= $\aleph(1)=n_1 ,...$, $\aleph(L-1)=n_{L-1})$\\
        \>$=Pr(E_0^{n_0},E_1^{n_1},...,E_{L-1}^{n_{L-1}})$\\
        \>$=Pr(E_0^{n_0})\cdot Pr(E_1^{n_1},...,E_{L-1}^{n_{L-1}}|E_0^{n_0})$\\
        \>$=Pr(E_0^{n_0})\cdot Pr(E_1^{n_1})\cdot Pr(E_2^{n_2},...,E_{L-1}^{n_{L-1}}|E_0^{n_0},E_1^{n_1})$\\
        \>$\vdots$\\
        \>$=Pr(E_{L-1}^{n_{L-1}}|E_0^{n_0},E_1^{n_1},...,E_{L-1}^{n_{L-1}})\cdot\prod\limits^{L-2}_{i=0}Pr(E_i^{n_i})
        =1\cdot\prod\limits^{L-2}_{i=0}Pr(E_i^{n_i})$\\
        \>$=\prod\limits^{L-2}_{i=0}{\frac{(MN-\sum\limits^{i-1}_{j=0}{n_j})!}{n_i!(MN-\sum\limits^{i}_{j=0}{n_j})!}}
        \frac{(L-i-1)^{MN-\sum\limits^{i}_{j=0}n_j}}{(L-i)^{MN-\sum\limits^{i-1}_{j=0}n_j}}$\\
        \>$=\frac{MN!}{n_0!(MN-n_0)!}\frac{(L-1)^{MN-n_0}}{L^{MN}}\cdot\frac{(MN-n_0)!}{n_1!(MN-n_0-n_1)!}\frac{(L-2)^{MN-n_0-n_1}}{(L-1)^{MN-n_0}}\ldots
        \frac{(n_{L-2}+n_{L-1})!}{n_{L-2}!n_{L-1}!}\frac{(1)^{n_{\tiny L-1}}}{2^{n_{L-2}+n_{L-1}}}$\\
        \>$= \frac{MN!}{n_0!...n_{L-1}!}\frac{1}{L^{MN}}$
%
    \end{tabbing}
\end{proof}
\begin{myremark}
    Theorem \ref{TheMultinomial} can be directly used to calculate $H(Y)$. Because the joint distribution of $f({\tiny\aleph(0),\aleph(1),,...,\aleph(L-1)})$ is given and $\forall l\in\{0,1,...,L-1\},\exists P_l=\aleph(l)/MN$, the joint distribution of $f(P_0,P_1,...,P_{L-1}))$ is also known. Therefore, the distribution of information entropy of an image block $H(Y)=-\sum\limits^{L-1}_{l=0}P_l\log_2{P_l}$ can be also obtained.
\end{myremark}
\begin{mycol}
    If Y is a binary image, namely $L=2$, then Theorem \ref{TheMultinomial} degrades to
    \begin{center}
        $Pr(\aleph(0)=n_0$, $\aleph(1)=n_1)= {MN!}/({n_0!n_{1}!}{2^{MN}})$
    \end{center}
\end{mycol}
\begin{proof}
Straightforward.
\end{proof}

Because the statistics of interest are the mean and variance of $H(Y)$ but not its distribution, an easier way to derive them is to first find $h(P_l)$ in Eqn. (\ref{eqnLevelEntropy}) and then to calculate the mean and variance of $H(Y)$ via Eqn. (\ref{eqnSumEntropy}).

\begin{equation}
\label{eqnLevelEntropy}
    h(P_l) =-P_l\log_2P_l
\end{equation}

\begin{equation}
\label{eqnSumEntropy}
    H(Y)=\sum\limits^{L-1}_{l=0}h(P_l)
\end{equation}
Using Corollary \ref{colP(l)} and Eqn. (\ref{eqnLevelEntropy}), the following expectation values can be obtained
\begin{equation}
\label{eqnE[h(l)]}
    E[h(P_l)]
    =\sum\limits^{MN}_{n=0}\frac{n}{MN}\log_2\frac{MN}{n}\cdot\frac{MN!(L-1)^{MN-n}}{n!(MN-n)!L^{MN}}
\end{equation}
\begin{equation}
\label{eqnE[h(l)^2]}
    E[h(P_l)^2]
    =\sum\limits^{MN}_{n=0}(\frac{n}{MN}\log_2\frac{MN}{n})^2\cdot\frac{MN!(L-1)^{MN-n}}{n!(MN-n)!L^{MN}}
\end{equation}
\begin{equation}
\label{eqnE[h(l1l2)]}
    E[h(P_{l_1})h(P_{l_2})]
    =\sum\limits^{MN}_{n_1=0}
    \sum\limits^{MN-n_1}_{n_2=0}
    (\frac{n_1}{MN}\log_2\frac{MN}{n_1})(\frac{n_2}{MN}\log_2\frac{MN}{n_2})
    \cdot\frac{MN!(L-2)^{MN-n_1-n_2}}{n_1!n_2!(MN-n_1-n_2)!L^{MN}}
\end{equation}
Therefore, $H(Y)$'s first and second moments, mean $\mu_H$ and $E[H(Y)^2]$ , can be obtained via Eqns. (\ref{eqnMeanH[Y]}) and (\ref{eqnMeanH[Y]^2}), respectively. Finally, its variance $\sigma_H^2$ can be obtained via Eqn. (\ref{eqnVarH[Y]}).
\begin{equation}
\label{eqnMeanH[Y]}
    \mu_H=E[H(Y)]=E[\sum\limits^{L-1}_{l=0}h(P_l)]=\sum\limits^{L-1}_{l=0}E[h(P_l)]=L\cdot E[h(P_l)]
\end{equation}
\begin{eqnarray}
\label{eqnMeanH[Y]^2}
        E[H(Y)^2] &=&E[(\sum\limits^{L-1}_{l=0}h(P_l))^2]\nonumber \\
        &=&E[\sum\limits^{L-1}_{l=0}h(P_l)^2+\sum\limits^{L-1}_{l_1=0}\sum\limits^{L-1}_{l_2=0, l_2\neq l_1}h(P_{l_1})h(P_{l_2})]\\
        &=&L\cdot E[h(P_l)^2]+L(L-1)\cdot E[h(P_{l_1})h(P_{l_2})]\nonumber
\end{eqnarray}
\begin{eqnarray}
\label{eqnVarH[Y]}
    \sigma_H^2 &=&Var[H(Y)] =E[H(Y)^2]-E[H(Y)]^2 \nonumber \\
    &=&L\cdot E[h(P_l)^2]+L(L-1)\cdot E[h(P_{l_1})h(P_{l_2})]-L^2\cdot E[h(P_l)]^2
\end{eqnarray}
Therefore, an image block $Y$ at size $M$-by-$N$ within an ideally encrypted image $X$ with $L$ scales has the mean and variance of Shannon entropy as Eqns. (\ref{eqnMeanH[Y]}) and (\ref{eqnVarH[Y]}), respectively. It is clear that $\mu_H$ and $\sigma_H^2$ are determined as long as the size of $MN$ and the number of allowed scales $L$ are specified.

\begin{equation}
\label{eqnMeanHk}
    \mu_{H_K^*}=E[\overline{H^*_K}]=E[\sum\limits^{K}_{i=1}H(Y_i)]/K=\mu_H
\end{equation}
\begin{equation}
\label{eqnVarHk}
    \sigma_{H_K^*}^2=Var[\overline{H^*_K}]=Var[\sum\limits^{K}_{i=1}H(Y_i)/K]=\sigma_H^2/K
\end{equation}

Based on $\mu_H$ and $\sigma_H$, the theoretical mean and variance of $\overline{H^*_K}$ can be obtained as Eqns. (\ref{eqnMeanHk}) and (\ref{eqnVarHk}), respectively. Once a test encrypted image is given and the block entropy test is applied, the actual sample mean $\overline{H_K}$ is obtained. If the test image is ideally encrypted, then this actual sample mean should be comparable with the theoretical sample mean $\mu_{H^*_K}$, which is derived from the ideal case. Otherwise, the actual sample mean $\overline{H_K}$ should be lower than $\mu_{H^*_K}$.

\subsubsection{Numerical Mean and Variance of $\overline {H_K^*}$}
Numerically, the mean and variance of $H(Y)$ can be calculated via Eqns. (\ref{eqnMeanH[Y]}) and (\ref{eqnVarH[Y]}) as Table 1 shows. Consequently, the mean and variance of $\overline {H_K^*}$ are obtained via Eqns. (\ref{eqnMeanHk}) and (\ref{eqnVarHk}) according to the used number of sample blocks $K$.

\begin{table}[htbp]
\scriptsize
\caption{Numerical Results of $\mu_H$ and $\sigma_H$ of a image block encrypted by an ideal cipher}
\begin{center}
\begin{tabular}{|c|c|c|c|c|}
\hline
 & \multicolumn{ 2}{c|}{\textbf{Binary Image: $L=2$}} & \multicolumn{ 2}{c|}{\textbf{Gray Image: $L=256$}} \\ \hline
\textbf{\textit{size $M$-by-$N$}} & \textbf{$\mu_H$ } & \textbf{$\sigma_H$} & \textbf{$\mu_H$ } & \textbf{$\sigma_H$} \\ \hline
\textbf{$2$-by-$2$} & 0.7806390622 & 0.3077153752 & 1.9883002343 & 0.0760641186 \\ \hline
\textbf{$4$-by-$4$} & 0.9533616074 & 0.0661878066 & 3.9420646175 & 0.0828513507 \\ \hline
\textbf{$8$-by-$8$} & 0.9886389750 & 0.0160692363 & 5.7657169289 & 0.0766034388 \\ \hline
\textbf{$16$-by-$16$} & 0.9971767038 & 0.0039927774 & 7.1749663525 & 0.0524379986 \\ \hline
\textbf{$32$-by-$32$} & 0.9992952146 & 0.0009967175 & 7.8087565712 & 0.0172463431 \\ \hline
\hline
 & \multicolumn{ 2}{c|}{\textbf{Binary Image: $L=2$}} & \multicolumn{ 2}{c|}{\textbf{Gray Image: $L=256$}} \\ \hline
\textbf{\textit{size $M$-by-$N$}} & \textbf{$\mu_{H_K^*}$ } & \textbf{$\sigma_{H_K^*}$} & \textbf{$\mu_{H_K^*}$ } & \textbf{$\sigma_{H_K^*}$} \\ \hline
\textbf{$2$-by-$2$} & 0.7806390622 & 0.3077153752/$\sqrt{K}$ & 1.9883002343 & 0.0760641186/$\sqrt{K}$ \\ \hline
\textbf{$4$-by-$4$} & 0.9533616074 & 0.0661878066/$\sqrt{K}$ & 3.9420646175 & 0.0828513507/$\sqrt{K}$ \\ \hline
\textbf{$8$-by-$8$} & 0.9886389750 & 0.0160692363/$\sqrt{K}$ & 5.7657169289 & 0.0766034388/$\sqrt{K}$ \\ \hline
\textbf{$16$-by-$16$} & 0.9971767038 & 0.0039927774/$\sqrt{K}$ & 7.1749663525 & 0.0524379986/$\sqrt{K}$ \\ \hline
\textbf{$32$-by-$32$} & 0.9992952146 & 0.0009967175/$\sqrt{K}$ & 7.8087565712 & 0.0172463431/$\sqrt{K}$ \\ \hline
\end{tabular}
\end{center}
\label{tabNumH[l]}
\end{table}

From Table 1, it is clear that as the size of test block $MN$ increases, the expected block entropy increases quickly towards to its theoretical upper bound, while the variance of block entropy decreases quickly towards to zero. This conclusion is not surprising, because the intensity scale for each pixel is assumed to follow an uniform distribution ${\cal U}(0,L-1)$, so the more pixels in the test block, the higher the block entropy. Meanwhile, the intensity scale for each pixel is also assumed independently distributed, so the variance tends to decrease as the sample size increase. It is clear that if a test image has a high block entropy when the block size is small, it tends to have a higher entropy when this block size increases. However, the opposite claim does not always hold: a high entropy with a large block size is not necessary to imply a relative high entropy when this size decreases.

\begin{figure}[htb]
    \label{fig-GEntropy}
    \begin{minipage}[b]{0.23\linewidth}
      \centering
     \centerline{\includegraphics[width=\linewidth]{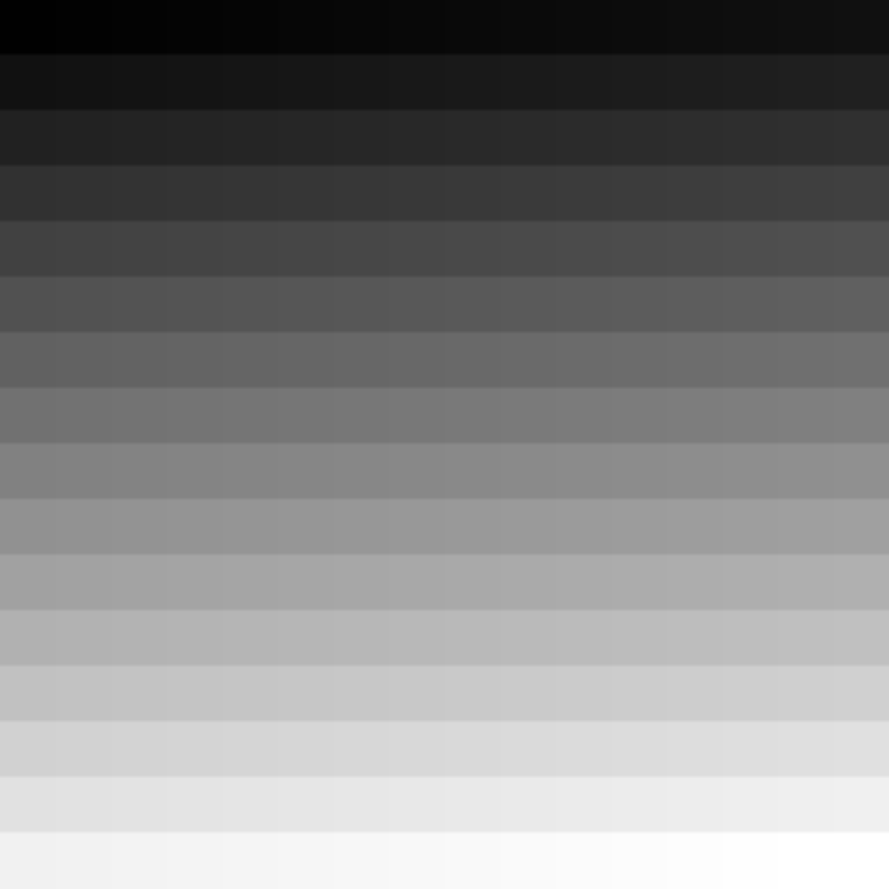}}
     \centerline{\small(a) Pattern}
    \end{minipage}\hfill
    \begin{minipage}[b]{0.23\linewidth}
      \centering
     \centerline{\includegraphics[width=\linewidth]{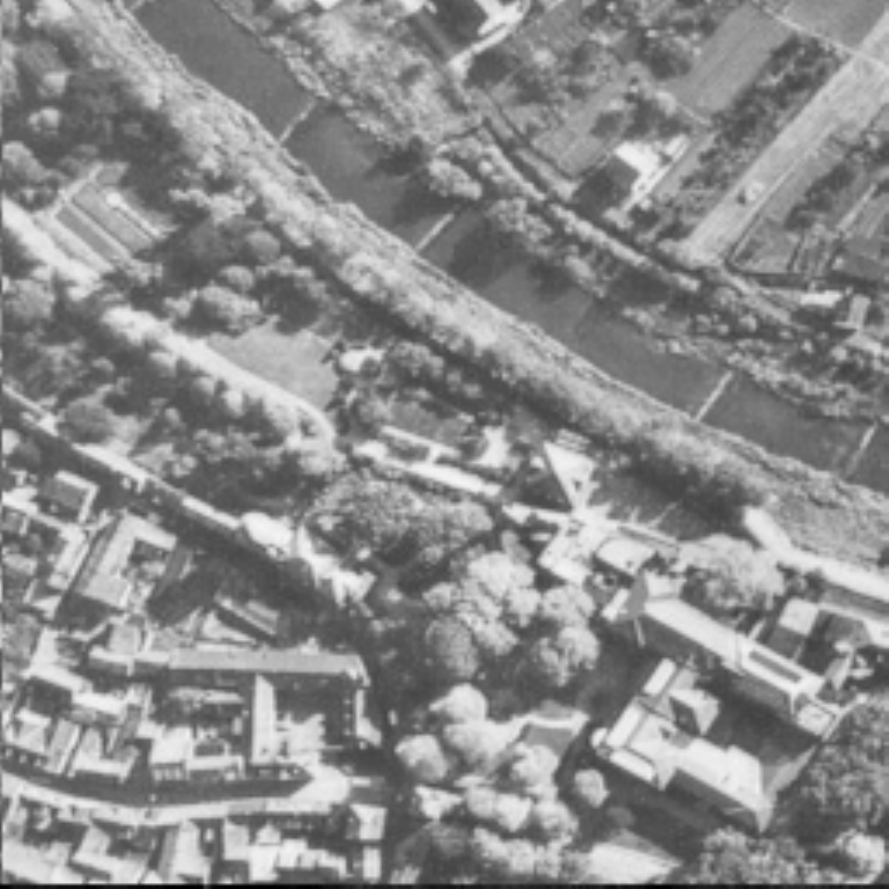}}
    \centerline{\small(b) Scene}
    \end{minipage}\hfill
    \begin{minipage}[b]{0.23\linewidth}
      \centering
     \centerline{\includegraphics[width=\linewidth]{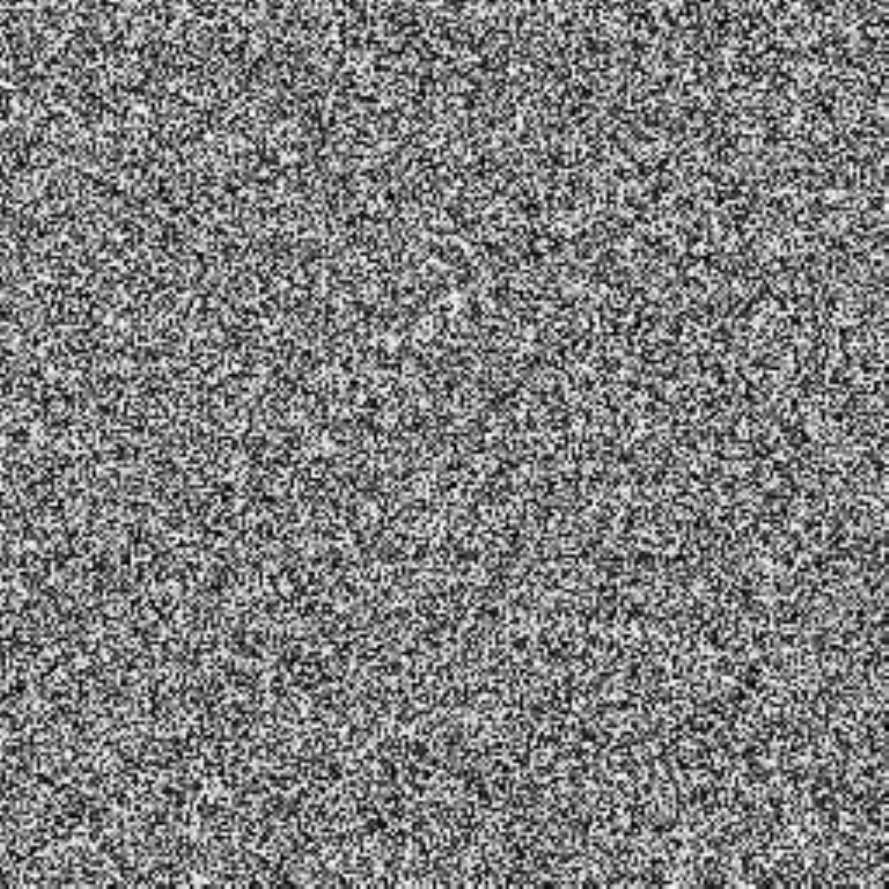}}
     \centerline{\small(c) Random}
    \end{minipage}\hfill
    \begin{minipage}[b]{0.27\linewidth}
      \centering
     \centerline{\includegraphics[width=\linewidth]{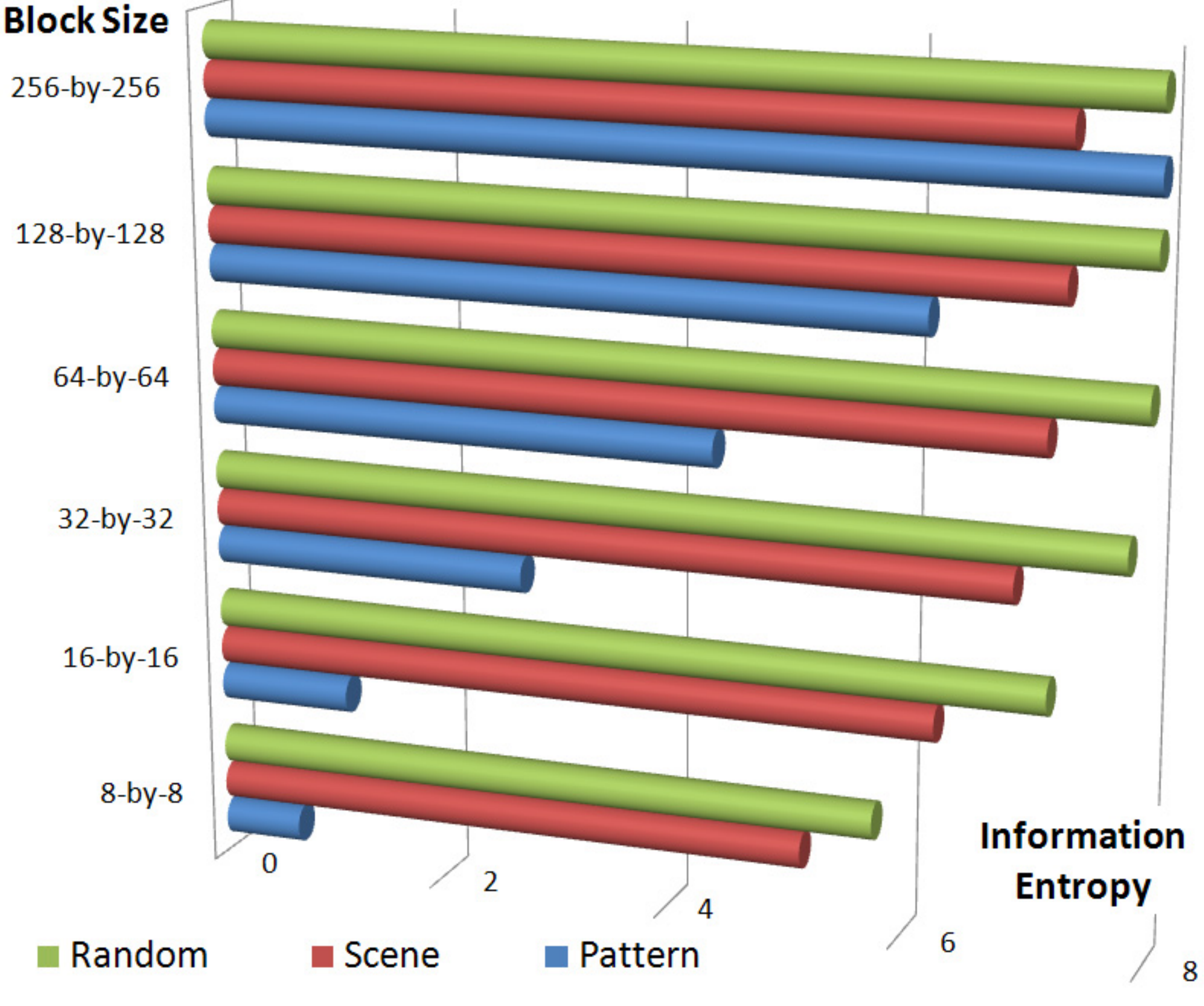}}
     \centerline{\small(d) Block Entropy Comparison}
    \end{minipage}\hfill
    \caption{Sample Images with Block Entropy Comparison}
\end{figure}
Fig. 3 shows three images, namely 'Pattern', 'Scene' and 'Random' with their block entropy results. It is obvious that the tendency of Shannon entropy is ascending as the block size increases. Although 'Pattern' reaches the upper bound of Shannon entropy at the block size of 256-by-256, its block entropy at smaller size is far lower than the other two test images. This result shows another example of why it is inappropriate to use global entropy for testing image randomness and thus for measuring image encryption quality.

In addition, from the comparison results in
Fig. 3-(d), it is noticeable that the block entropies of three test images get closer as the block size increases. On the other hand, from Table 1, too small block size may not be large enough to test the entropy. For example, a size of 2-by-2 block is obviously small for testing the entropy in a gray image, which has 256 intensity scales. Therefore, in order to obtain a good test results, an appropriate block size has to be determined.

Heuristically, we believed that $MN=256$ is a good block size for gray image $L=256$. 
In the following section, the Z-test on $\overline{H_K}$ is derived for $L=256$ and $MN=256$.
\subsection{Z-Test on $\overline {H_K}$}
Conventional quantitative measurements are commonly not easy to follow. For example, how flat a histogram of the ciphertext is 'flat' enough? Is an entropy of 7.9999 close enough to the theoretical upper bound 8 or is 7.9911 the value that should be believed to be close enough? Therefore, it is desired to make a qualitative measurement rather than a quantitative for these cases. This section introduces a qualitative measurement for image randomness using the Z-test.

With results derived from an ideally encrypted image, we know that the sample mean $\overline {H_K^*}$ of Shannon entropy on $K$ image blocks follows the normal distribution ${\cal N}(\mu_{H_K^*},\sigma_{H_K^*}^2)$. However, in reality, the actual sample mean $\overline {H_K}$ may or may not follow the ideal distribution.
Without loss of generality, suppose the true mean of the block entropy for a given image is $\mu_{test}$. Therefore, the hypotheses, then, are:
\begin{itemize}
  \item Null Hypothesis ${\cal H}_0$: $\mu_{test}=\mu_{H_K^*}$ (the test image $X$ is ideally encrypted)
  \item Alternative Hypothesis ${\cal H}_1$: $\mu_{test}<\mu_{H_K^*}$ (the test image $X$ is not ideally encrypted)
\end{itemize}
Define the Z statistics for block entropy test as Eqn. (\ref{eqnZtestImg}) shows.
\begin{equation}
\label{eqnZtestImg}
    Z_K = \frac{\overline{H_K}-\mu_{H_K^*}}{\sigma_{H_K^*}}
\end{equation}

We interested in two things: (1) the critical value $H_c$ such that $\forall \overline{H_K}< H_c$ we reject ${\cal H}_0$ and $\forall \overline{H_K}\geq H_c$ we accept ${\cal H}_0$; and (2) the error rate of rejecting ${\cal H}_0$ when ${\cal H}_0$ is true, which is called Type I error with the definition in Eqn. (\ref{eqnAlpha}). It is clear that these two things are closely related for the smaller $H_c$ is used the smaller $\alpha$ is. Statistics tells us that given an $\alpha$ level of significance, the critical value $H_c$ can be found via Eqn. (\ref{eqnCutoff}), where $K$ is sufficiently large.
\begin{equation}
\label{eqnAlpha}
\alpha = Pr(\textit{  reject } {\cal H}_0 | {\cal H}_0 \textit{  is True})
\end{equation}
\begin{equation}
\label{eqnCutoff}
    H_c = \mu_{H_K^*}-\Phi^{-1}(\alpha)\cdot\sigma_{H_K^*}
\end{equation}
For any given $\alpha$, a corresponding critical value $H_c$ can be calculated via Eqn. (\ref{eqnCutoff}). Table \ref{tabCriticalI} shows critical values for deciding whether a test image is ideally encrypted or not when the block size 16-by-16 and $L = 256$.

\begin{table}[htbp]
\scriptsize
\caption{Critical Values of Hypothesis Test for Image Encryption ($MN=256$ and $L=256$)}
\begin{center}
\begin{tabular}{|c|c|c|c|c|c|}
\hline
\textbf{} & \textbf{$K=36$} & \textbf{$K=49$} & \textbf{$K=64$} & \textbf{$K=81$} & \textbf{$K=100$} \\ \hline
\textbf{$\alpha = 0.05$} & 7.1605908805 & 7.1626445194 & 7.1641847485 & 7.1653827045 & 7.1663410693 \\ \hline
\textbf{$\alpha = 0.01$} & 7.1546348481 & 7.1575393487 & 7.1597177242 & 7.1614120162 & 7.1627674499 \\ \hline
\textbf{$\alpha = 0.001$} & 7.147958753 & 7.1518169815 & 7.1547106529 & 7.1569612862 & 7.1587617928 \\ \hline\hline
\textbf{} & \textbf{$K=121$} & \textbf{$K=144$} & \textbf{$K=169$} & \textbf{$K=196$} & \textbf{$K=225$} \\ \hline
\textbf{$\alpha = 0.05$} & 7.167125186 & 7.1677786165 & 7.1683315193 & 7.1688054359 & 7.1692161637 \\ \hline
\textbf{$\alpha = 0.01$} & 7.163876441 & 7.1648006003 & 7.1655825813 & 7.1662528506 & 7.1668337508 \\ \hline
\textbf{$\alpha = 0.001$} & 7.1602349346 & 7.1614625528 & 7.1625013066 & 7.163391667 & 7.1641633127 \\ \hline\hline
\textbf{} & \textbf{$K=256$} & \textbf{$K=289$} & \textbf{$K=324$} & \textbf{$K=361$} & \textbf{$K=400$} \\ \hline
\textbf{$\alpha = 0.05$} & 7.1695755505 & 7.1698926565 & 7.1701745285 & 7.1704267298 & 7.1706537109 \\ \hline
\textbf{$\alpha = 0.01$} & 7.1673420384 & 7.1677905274 & 7.1681891844 & 7.1685458774 & 7.1688669012 \\ \hline
\textbf{$\alpha = 0.001$} & 7.1648385027 & 7.1654342586 & 7.1659638193 & 7.1664376369 & 7.1668640727 \\ \hline
\end{tabular}
\end{center}
\label{tabCriticalI}
\end{table}

It is worth while to note that the Eqn. (\ref{eqnCutoff}) holds only for a sufficiently large $K$. In other words, if the CLT is not satisfied, then Eqn. (\ref{eqnCutoff}) is inaccurate. Therefore, the last question for the proposed hypothesis test is when $K$ is sufficiently large such that the CLT can be applied. Unfortunately, this is a controversial question. Many statisticians believe that if the sample size is larger than 30 then $K$ is sufficiently large \cite{BarronStatistics}\cite{StatisticsAndData}, but some others suggest other sample sizes \cite{Statistics}, for example 100 \cite{Probability_Statistics}. If the sample size $K$ is considered to be sufficiently large, the contents in Table \ref{tabCriticalI} can be directly used.


A safer way to check how good or bad critical values in Table \ref{tabCriticalI} are is to apply the BET, which gives error bounds of approximating the sample mean with the standard normal distribution. Recall Eqn. (\ref{Korolev}), the value of $\rho/\sigma^3$ in our problem is found about 1.6 by using the Monte Carlo simulation for 100,000 trails. Substitute this value and the upper bound of $\cal C$, Eqn. (\ref{eqnBEsimplified}) is obtained.
\begin{equation}
\label{eqnBEsimplified}
|F_{Z_k}(z)-\Phi(z)|\leq 0.76544/\sqrt{K}
\end{equation}
With the BET, even if the sample size $K$ is not large enough to apply the CLT, it is still possible to use the CLT conclusion to estimate the sample mean. Because for any given value $H_c$ associated with an $\alpha$ in Table \ref{tabCriticalI}, Eqn. (\ref{eqnTable2}) always holds. As a result, the true type I error $\gamma$ of the hypothesis test for image encryption is obtained in Eqn. (\ref{eqnTable3}) .
\begin{equation}
\label{eqnTable2}
    \Phi(z<\frac{H_c-\mu_{H_K^*}}{\sigma_{H_K^*}}) = \alpha
\end{equation}
\begin{equation}
\label{eqnTable3}
\gamma =F_{Z_k}(z<\frac{H_c-\mu_{H_K^*}}{\sigma_{H_K^*}})\leq \frac{0.76544}{\sqrt{K}}+\Phi(z<\frac{H_c-\mu_{H_K^*}}{\sigma_{H_K^*}})=\frac{0.76544}{\sqrt{K}}+\alpha
\end{equation}
Table \ref{tabError} shows the $\gamma$ values corresponding to the critical values in Tabel \ref{tabCriticalI}. From the results, it is clear that the upper bound of Type I error when approximating the sample mean of block entropies with a normal distribution gets smaller as the sample size increases. Moreover, this upper bound tells how bad the proposed hypothesis test could be. For example, the critical value for $\alpha = 0.01$ and $K=100$ is $H_c = 7.1627674499$ in Tabel \ref{tabCriticalI}. The corresponding $\gamma$ for $\alpha = 0.01$ and $K=100$ is $\gamma = 0.08654$. Therefore, in the worst case, the possibility to reject ${\cal H}_0$ when it is true is 0.08654. It is worth to note that 0.08654 is the upper bound, which may not be actually reached. More likely, the CLT is applicable to this case and thus the error rate is 0.01.

\begin{table}[htbp]
\scriptsize
\caption{Upper Bounds of Type I Error $\gamma$ for the Block Entropy Test}
\begin{center}
\begin{tabular}{|c|c|c|c|c|c|}
\hline
\textbf{} & \textbf{$K=36$} & \textbf{$K=49$} & \textbf{$K=64$} & \textbf{$K=81$} & \textbf{$K=100$} \\ \hline
\textbf{$\alpha = 0.05$}  & 0.17757 & 0.15935 & 0.14568 & 0.13505 & 0.12654 \\ \hline
\textbf{$\alpha = 0.01$}  & 0.13757 & 0.11935 & 0.10568 & 0.09505 & 0.08654 \\ \hline
\textbf{$\alpha = 0.001$} & 0.12857 & 0.11035 & 0.09668 & 0.08605 & 0.07754 \\ \hline\hline
\textbf{} & \textbf{$K=121$} & \textbf{$K=144$} & \textbf{$K=169$} & \textbf{$K=196$} & \textbf{$K=225$} \\ \hline
\textbf{$\alpha = 0.05$}  & 0.11959 & 0.11379 & 0.10888 & 0.10467 & 0.10103 \\ \hline
\textbf{$\alpha = 0.01$} & 0.07959 & 0.07379 & 0.06888 & 0.06467 & 0.06103 \\ \hline
\textbf{$\alpha = 0.001$}  & 0.07059 & 0.06479 & 0.05988 & 0.05567 & 0.05203 \\ \hline\hline
\textbf{} & \textbf{$K=256$} & \textbf{$K=289$} & \textbf{$K=324$} & \textbf{$K=361$} & \textbf{$K=400$} \\ \hline
\textbf{$\alpha = 0.05$} & 0.09784 & 0.09503 & 0.09252 & 0.09029 & 0.08827 \\ \hline
\textbf{$\alpha = 0.01$}  & 0.05784 & 0.05503 & 0.05252 & 0.05029 & 0.04827 \\ \hline
\textbf{$\alpha = 0.001$} & 0.04884 & 0.04603 & 0.04352 & 0.04129 & 0.03927 \\ \hline
\end{tabular}
\end{center}
\label{tabError}
\end{table}
\section{Simulation Results}
Previous sections have already set up the block entropy test for image encryption. In this section, it is applied to measure the randomness of encrypted images both quantitatively and qualitatively.
\subsection{Measure the Quality of Image Shuffling}
Conventional image shuffling algorithms are very popular in encryption community for their compactness and fast encryption rate. Although simply shuffling pixels is insecure from the view point of cryptanalysis, the objective of this example is show how to use the block entropy test for measuring the quality of image shuffling, which is an impossible task for the global entropy test.

Generally speaking, a pixel-wise shuffling algorithm for image encryption is defined in Eqn. (\ref{eqnShuffling1}), where $S(i,j)$ indicates the pixel located at intersection of the $i$th row and the $j$th column in the shuffled image $S$; $e_\Pi^{rc}$ is a permutation for the row number and the column number. If additionally Eqn. (\ref{eqnShufflingRC}) holds, i.e. $e_\Pi^{rc}$ is a function separable with respect to the row variable and the column variable, then a pixel-wise shuffling degrades to a row-column-wise shuffling. As its definition implies, a shuffling algorithm does not change the intensity level for any pixel, but shuffles the positions of pixels. A shuffling algorithm can also be block cipher, which shuffles a certain size of image block at one time for fast computation.
\begin{eqnarray}
\label{eqnShuffling1}
S(i,j) = I(e_\Pi^{rc}(i,j))
\end{eqnarray}
\begin{eqnarray}
\label{eqnShufflingRC}
[e_\Pi^r(i),e_\Pi^c(j)] = e_\Pi^{rc}(i,j)
\end{eqnarray}

Fig. 4 shows examples of image encryption using shuffling algorithms. For simplicity, a shuffled image $I$ using pixel-wise shuffling scheme with a $M$-by-$N$ block size
is denoted as ${\cal S}^{M\times N}_{p-w}(I)$. Similarly, a shuffled image using the row-column-wise scheme 
is denoted as ${\cal S}^{M\times N}_{r-c-w}(I)$. The original image is the \textit{binary} logo image of Tufts University at size of 256-by-256. It is noticeable that the row-column-wise shuffled images have a mesh-like pattern because the pixels are shuffled with respect to a row-wise shuffling followed by a column-wise shuffling. In other words, pixels in a row/column are still in a row/clumn after shuffling and the only difference is that the row/column number of these pixels are changed after shuffling. From this series of images, it is noticeable that from the point view of human visual inspection:
\begin{itemize}
    \item ${\cal S}^{M\times N}(I)$ gets more random-like as $M\times N$ increases
    \item ${\cal S}^{M\times N}_{p-w}(I)$ is more random-like than ${\cal S}^{M\times N}_{r-c-w}(I)$
\end{itemize}

\begin{figure}[htbp]
    \label{fig-Shuffling}
    \begin{minipage}[b]{.99\linewidth}
      \centering
     \centerline{\includegraphics[width=3cm]{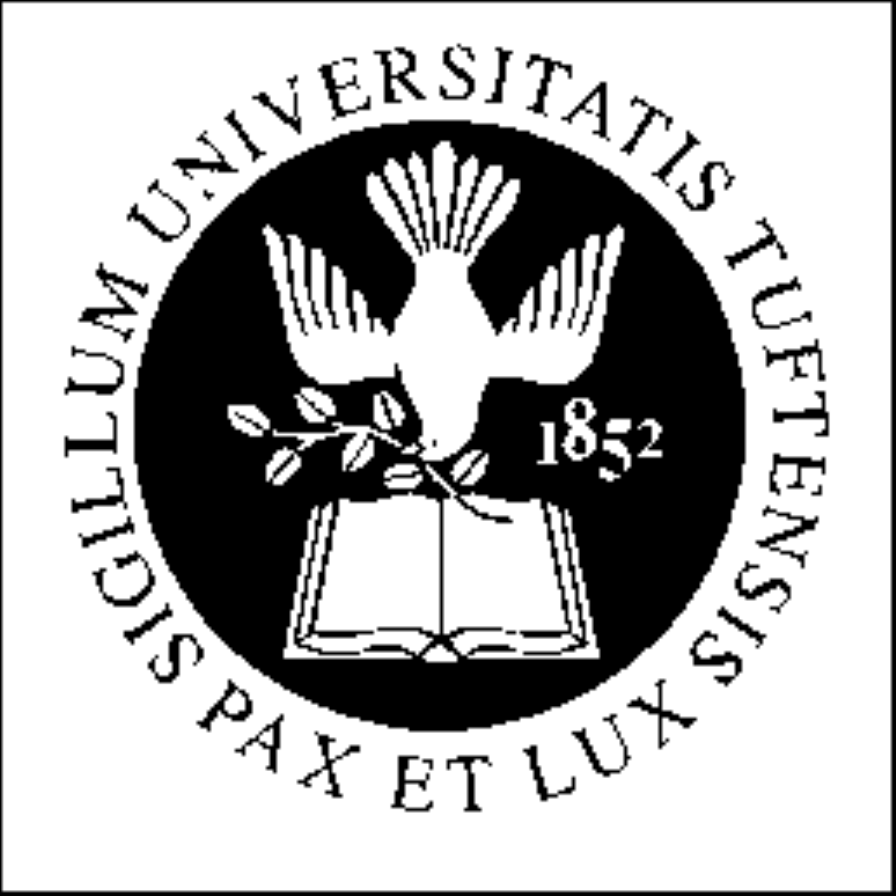}}
     \centerline{\scriptsize (a) Original Image I}
      \vspace{0.12cm}
    \end{minipage}\hfill
    \begin{minipage}[b]{0.22\linewidth}
      \centering
     \centerline{\includegraphics[width=3cm]{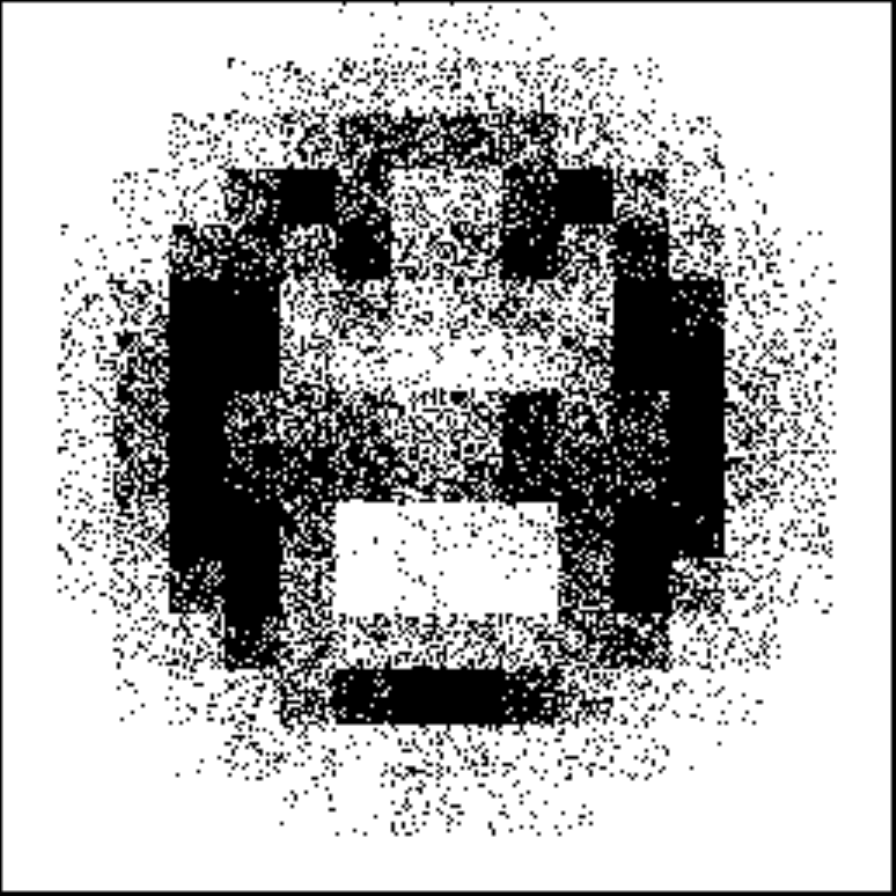}}
     \centerline{\scriptsize (b) ${\cal S}^{16\times 16}_{p-w}(I)$}
      \vspace{0.12cm}
    \end{minipage}\hfill
    \begin{minipage}[b]{.22\linewidth}
      \centering
     \centerline{\includegraphics[width=3cm]{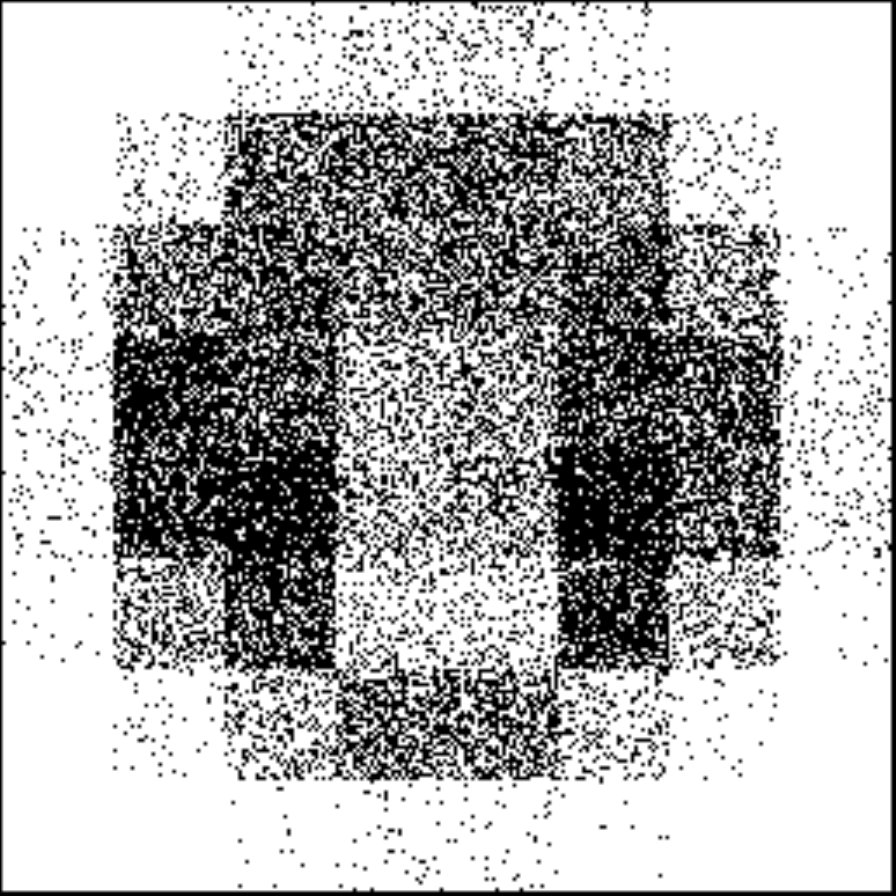}}
     \centerline{\scriptsize (c) ${\cal S}^{32\times 32}_{p-w}(I)$}
      \vspace{0.12cm}
    \end{minipage}\hfill
    \begin{minipage}[b]{.22\linewidth}
      \centering
     \centerline{\includegraphics[width=3cm]{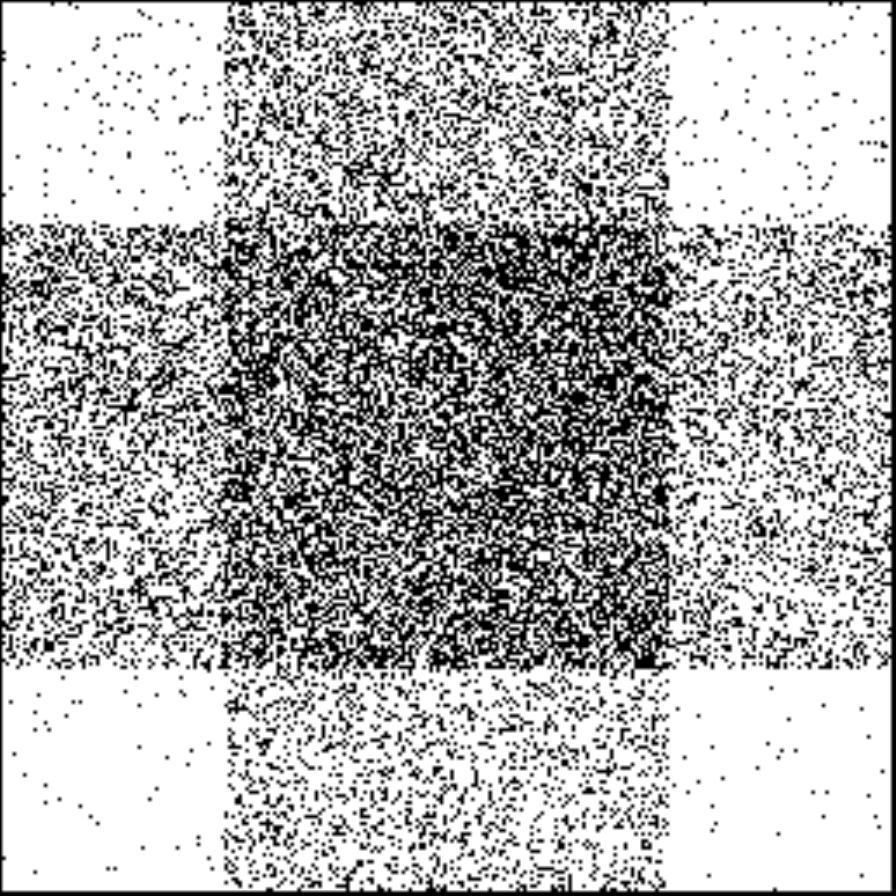}}
     \centerline{\scriptsize (d) ${\cal S}^{64\times 64}_{p-w}(I)$}
      \vspace{0.12cm}
    \end{minipage}\hfill
    \begin{minipage}[b]{.22\linewidth}
      \centering
     \centerline{\includegraphics[width=3cm]{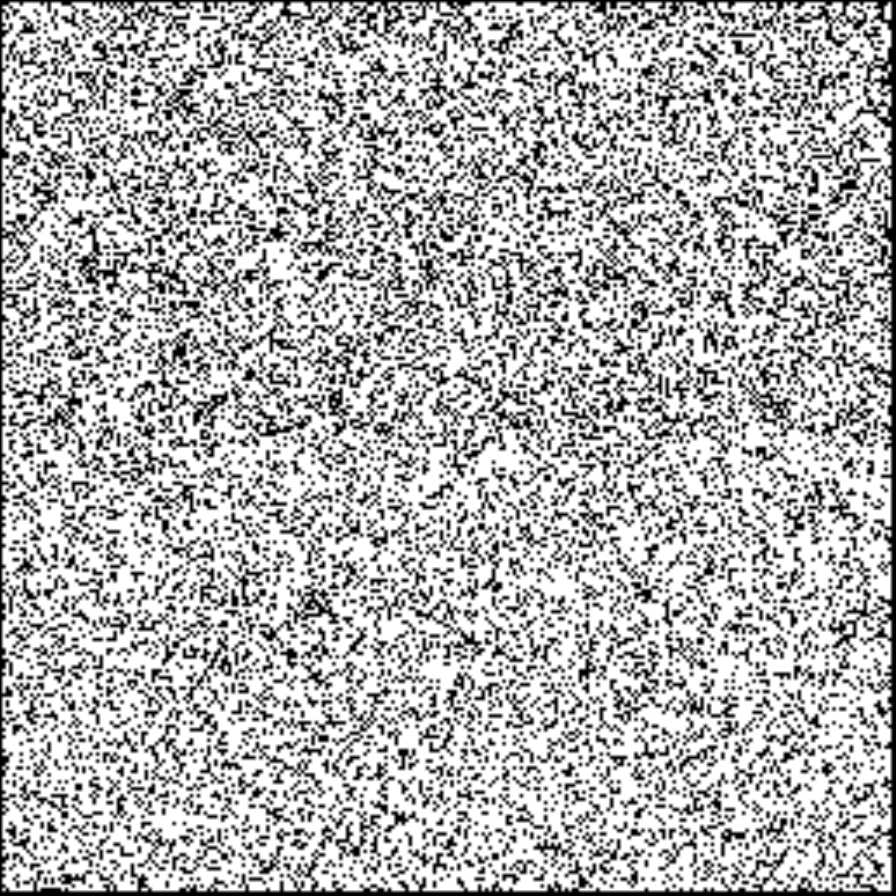}}
     \centerline{\scriptsize (e) ${\cal S}^{128\times 128}_{p-w}(I)$}
      \vspace{0.12cm}
    \end{minipage}\hfill
       \begin{minipage}[b]{0.22\linewidth}
      \centering
     \centerline{\includegraphics[width=3cm]{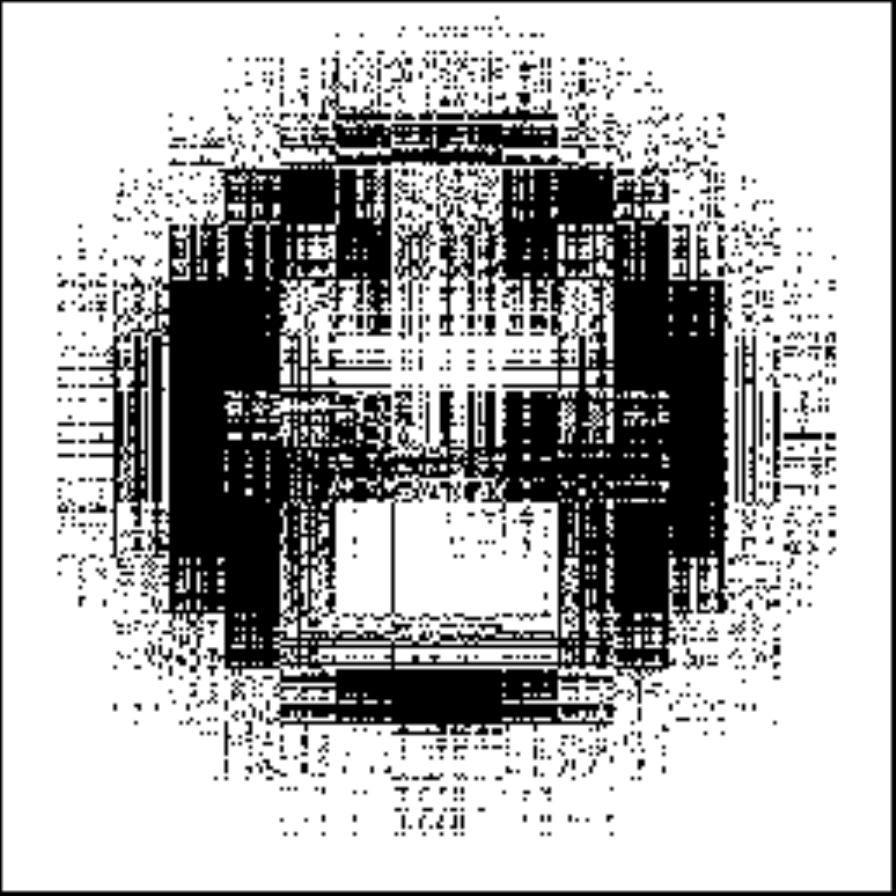}}
     \centerline{\scriptsize (f) ${\cal S}^{16\times 16}_{r-c-w}(I)$}
      \vspace{0.12cm}
    \end{minipage}\hfill
    \begin{minipage}[b]{.22\linewidth}
      \centering
     \centerline{\includegraphics[width=3cm]{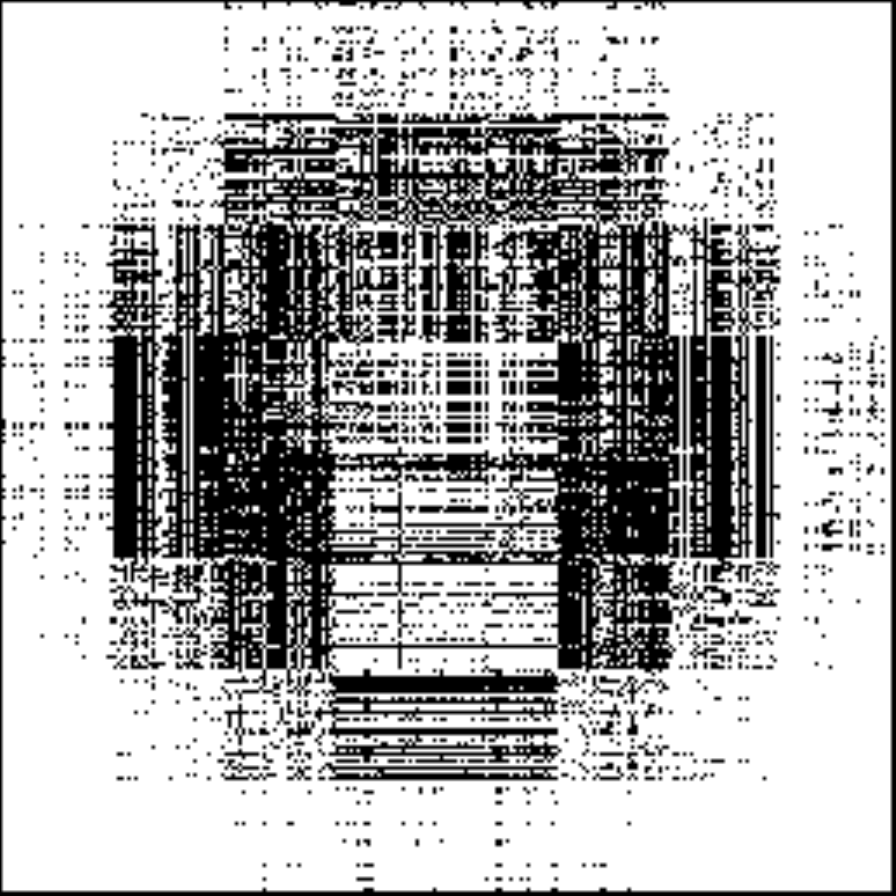}}
     \centerline{\scriptsize (g) ${\cal S}^{32\times 32}_{r-c-w}(I)$}
      \vspace{0.12cm}
    \end{minipage}\hfill
    \begin{minipage}[b]{.22\linewidth}
      \centering
     \centerline{\includegraphics[width=3cm]{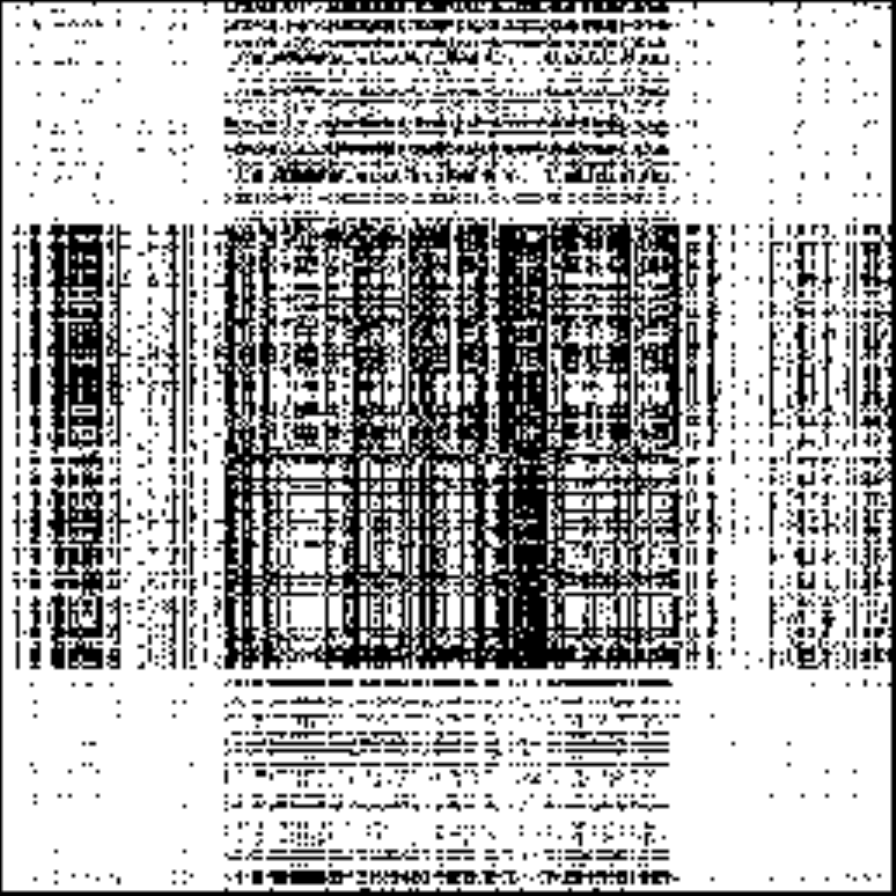}}
     \centerline{\scriptsize (h) ${\cal S}^{64\times 64}_{r-c-w}(I)$}
      \vspace{0.12cm}
    \end{minipage}\hfill
    \begin{minipage}[b]{.22\linewidth}
      \centering
     \centerline{\includegraphics[width=3cm]{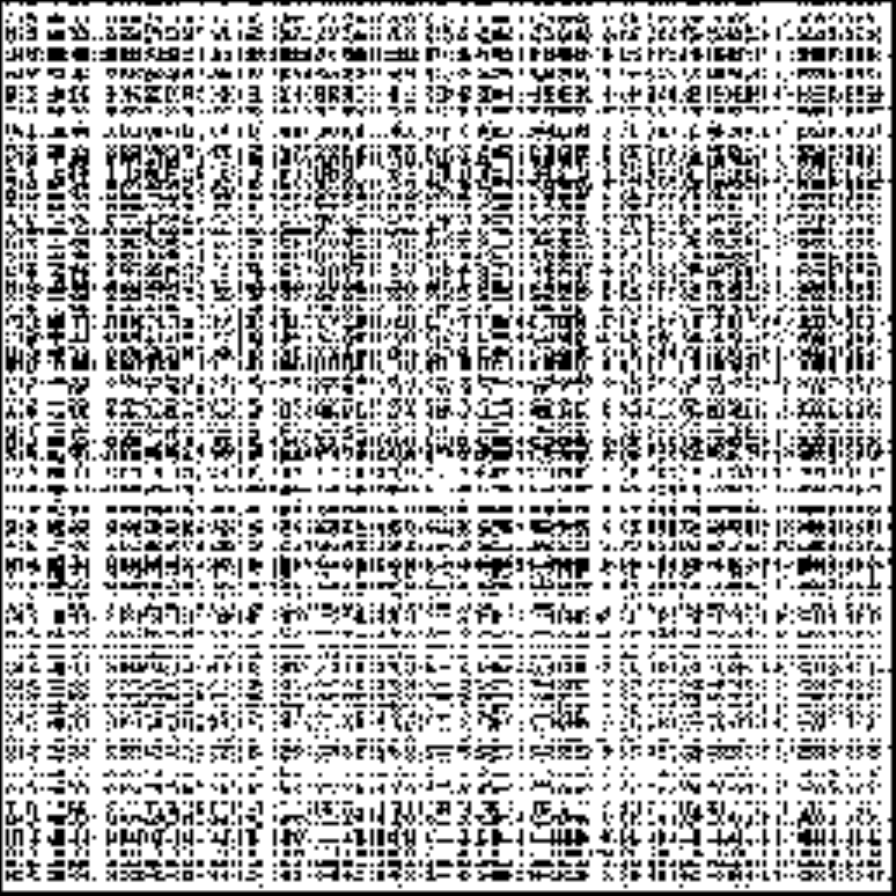}}
     \centerline{\scriptsize (i) ${\cal S}^{128\times 128}_{r-c-w}(I)$}
      \vspace{0.12cm}
    \end{minipage}\hfill
    \caption{Image shuffling for the binary Tufts logo}
\end{figure}

It is clear that testing the Shannon entropy with respect to the global image is ill-posed to measure the randomness of a shuffling based encryption algorithm,which does not change the image global statistics. Alternatively, the FIPS 140-2 tests for binary sequences and the block entropy test are used to test the randomness of images in Fig. 4.

\begin{table}[htp]
\centering
\caption{Randomness Test for Binary Images in Fig. 4}
{\scriptsize
\begin{center}
\begin{tabular}{|c|c|c|c|c|c|c|c|c|c|c|c|c|}
\hline
\multicolumn{ 1}{|c|}{\textbf{Comparison}} & \multicolumn{ 10}{c|}{\textbf{FIPS 140-2 Tests}} & \multicolumn{ 2}{c|}{\textbf{Block Entropy Test}} \\ \cline{ 2- 11}
\multicolumn{ 1}{|c|}{} & \multicolumn{ 2}{c|}{\textbf{Monobit}} & \multicolumn{ 1}{c|}{\textbf{Poker}} & \multicolumn{ 6}{c|}{\textbf{Run}} & \textbf{Long  } & \multicolumn{ 2}{c|}{\textbf{}} \\ \cline{ 5- 10}
\multicolumn{ 1}{|c|}{} & \multicolumn{ 2}{c|}{\textbf{}} & \multicolumn{ 1}{c|}{} & \multicolumn{ 6}{c|}{\textbf{Length of The Run}} & \textbf{} & \multicolumn{ 2}{c|}{\textbf{}} \\ \cline{ 5- 10} \cline{12-13}
\multicolumn{ 1}{|c|}{} & \multicolumn{ 2}{c|}{\textbf{}} & \multicolumn{ 1}{c|}{} & \textbf{1} & \textbf{2} & \textbf{3} & \textbf{4} & \textbf{5} & \textbf{$>=6$} & \textbf{Run} & \textbf{Mean} & \textbf{Std} \\ \hline
\textbf{\tiny Theoretical Value} & \multicolumn{ 2}{c|} {\textit{9725-10725}} & \textit{2.16-46.17} & \textit{2315-2685} & \textit{1114-1386} & \textit{527-723} & \textit{240-384} & \textit{103-209} & \textit{103-209} & \textit{0} & \textit{0.9971767} & \textit{0.0003993} \\ \hline

\multicolumn{ 1}{|c|}{\textbf{\tiny $I$}} & \textit{0 bit} & 11299 & \multicolumn{ 1}{c|}{2913.984} & 106 & 131 & 109 & 36 & 42 & 280 & \multicolumn{ 1}{c|}{110} & \multicolumn{ 1}{c|}{0.4576711} & \multicolumn{ 1}{c|}{0.3697290} \\ \cline{ 2- 3}\cline{ 5- 10}
\multicolumn{ 1}{|c|}{} & \textit{1 bit} & 8701 & \multicolumn{ 1}{c|}{} & 53 & 70 & 59 & 54 & 73 & 395 & \multicolumn{ 1}{c|}{} & \multicolumn{ 1}{c|}{} & \multicolumn{ 1}{c|}{} \\ \hline
\multicolumn{ 1}{|c|}{\textbf{\tiny ${\cal S}^{16\times 16}_{p-w}(I)$}} & \textit{0 bit} & 11244 & \multicolumn{ 1}{c|}{2386.291} & 1371 & 530 & 292 & 137 & 101 & 285 & \multicolumn{ 1}{c|}{82} & \multicolumn{ 1}{c|}{0.5498039} & \multicolumn{ 1}{c|}{0.3324658} \\ \cline{ 2- 3}\cline{ 5- 10}
\multicolumn{ 1}{|c|}{} & \textit{1 bit} & 8756 & \multicolumn{ 1}{c|}{} & 1238 & 565 & 309 & 169 & 84 & 352 & \multicolumn{ 1}{c|}{} & \multicolumn{ 1}{c|}{} & \multicolumn{ 1}{c|}{} \\ \hline
\multicolumn{ 1}{|c|}{\textbf{\tiny ${\cal S}^{32\times 32}_{p-w}(I)$}} & \textit{0 bit} & 11251 & \multicolumn{ 1}{c|}{695.8784} & 1774 & 830 & 499 & 239 & 184 & 461 & \multicolumn{ 1}{c|}{7} & \multicolumn{ 1}{c|}{0.6660959} & \multicolumn{ 1}{c|}{0.3168008} \\ \cline{ 2- 3}\cline{ 5- 10}
\multicolumn{ 1}{|c|}{} & \textit{1 bit} & 8749 & \multicolumn{ 1}{c|}{} & 2232 & 828 & 370 & 199 & 113 & 246 & \multicolumn{ 1}{c|}{} & \multicolumn{ 1}{c|}{} & \multicolumn{ 1}{c|}{} \\ \hline
\multicolumn{ 1}{|c|}{\textbf{\tiny ${\cal S}^{64\times 64}_{p-w}(I)$}} & \textit{0 bit} & 9411 & \multicolumn{ 1}{c|}{869.2224} & 2354 & 1050 & 512 & 263 & 155 & 218 & \multicolumn{ 1}{c|}{0} & \multicolumn{ 1}{c|}{0.7723314} & \multicolumn{ 1}{c|}{0.2896868} \\ \cline{ 2- 3}\cline{ 5- 10}
\multicolumn{ 1}{|c|}{} & \textit{1 bit} & 10589 & \multicolumn{ 1}{c|}{} & 2105 & 1112 & 550 & 298 & 165 & 323 & \multicolumn{ 1}{c|}{} & \multicolumn{ 1}{c|}{} & \multicolumn{ 1}{c|}{} \\ \hline
\multicolumn{ 1}{|c|}{\textbf{\tiny ${\cal S}^{128\times 128}_{p-w}(I)$}} & \textit{0 bit} & 5872 & \multicolumn{ 1}{c|}{4582.874} & 2897 & 856 & 259 & 79 & 26 & 6 & \multicolumn{ 1}{c|}{0} & \multicolumn{ 1}{c|}{0.8728581} & \multicolumn{ 1}{c|}{0.0552799} \\ \cline{ 2- 3}\cline{ 5- 10}
\multicolumn{ 1}{|c|}{} & \textit{1 bit} & 14128 & \multicolumn{ 1}{c|}{} & 1174 & 862 & 592 & 458 & 315 & 722 & \multicolumn{ 1}{c|}{} & \multicolumn{ 1}{c|}{} & \multicolumn{ 1}{c|}{} \\ \hline
\multicolumn{ 1}{|c|}{\textbf{\tiny ${\cal S}^{16\times 16}_{r-c-w}(I)$}} & \textit{0 bit} & 11235 & \multicolumn{ 1}{c|}{2802.157} & 1137 & 348 & 146 & 114 & 26 & 222 & \multicolumn{ 1}{c|}{97} & \multicolumn{ 1}{c|}{0.5399521}
& \multicolumn{ 1}{c|}{0.3266322} \\ \cline{ 2- 3}\cline{ 5- 10}
\multicolumn{ 1}{|c|}{} & \textit{1 bit} & 8765 & \multicolumn{ 1}{c|}{} & 857 & 481 & 190 & 122 & 43 & 301 & \multicolumn{ 1}{c|}{} & \multicolumn{ 1}{c|}{} & \multicolumn{ 1}{c|}{} \\ \hline
\multicolumn{ 1}{|c|}{\textbf{\tiny ${\cal S}^{32\times 32}_{r-c-w}(I)$}} & \textit{0 bit} & 11179 & \multicolumn{ 1}{c|}{822.0224} & 1457 & 476 & 212 & 138 & 67 & 339 & \multicolumn{ 1}{c|}{83} & \multicolumn{ 1}{c|}{0.6431998} & \multicolumn{ 1}{c|}{0.3261251} \\ \cline{ 2- 3}\cline{ 5- 10}
\multicolumn{ 1}{|c|}{} & \textit{1 bit} & 8821 & \multicolumn{ 1}{c|}{} & 1259 & 610 & 157 & 190 & 119 & 355 & \multicolumn{ 1}{c|}{} & \multicolumn{ 1}{c|}{} & \multicolumn{ 1}{c|}{} \\ \hline
\multicolumn{ 1}{|c|}{\textbf{\tiny ${\cal S}^{64\times 64}_{r-c-w}(I)$}} & \textit{0 bit} & 9917 & \multicolumn{ 1}{c|}{377.1584} & 1535 & 600 & 302 & 115 & 112 & 348 & \multicolumn{ 1}{c|}{57} & \multicolumn{ 1}{c|}{0.7548611} & \multicolumn{ 1}{c|}{0.3052360} \\ \cline{ 2- 3}\cline{ 5- 10}
\multicolumn{ 1}{|c|}{} & \textit{1 bit} & 10083 & \multicolumn{ 1}{c|}{} & 1476 & 574 & 338 & 109 & 139 & 375 & \multicolumn{ 1}{c|}{} & \multicolumn{ 1}{c|}{} & \multicolumn{ 1}{c|}{} \\ \hline
\multicolumn{ 1}{|c|}{\textbf{\tiny ${\cal S}^{128\times 128}_{r-c-w}(I)$}} & \textit{0 bit} & 5370 & \multicolumn{ 1}{c|}{10966.144} & 1957 & 714 & 384 & 79 & 56 & 34 & \multicolumn{ 1}{c|}{0} & \multicolumn{ 1}{c|}{0.8508929} & \multicolumn{ 1}{c|}{0.1121833} \\ \cline{ 2- 3}\cline{ 5- 10}
\multicolumn{ 1}{|c|}{} & \textit{1 bit} & 14630 & \multicolumn{ 1}{c|}{} & 1075 & 752 & 449 & 269 & 156 & 524 & \multicolumn{ 1}{c|}{} & \multicolumn{ 1}{c|}{} & \multicolumn{ 1}{c|}{} \\ \hline
\end{tabular}
\end{center}}
\label{Randomness Table}
\end{table}

\begin{figure}[hbp]
    \label{fig-FIPS}
    \begin{minipage}[b]{.48\linewidth}
      \centering
     \centerline{\includegraphics[width=\linewidth]{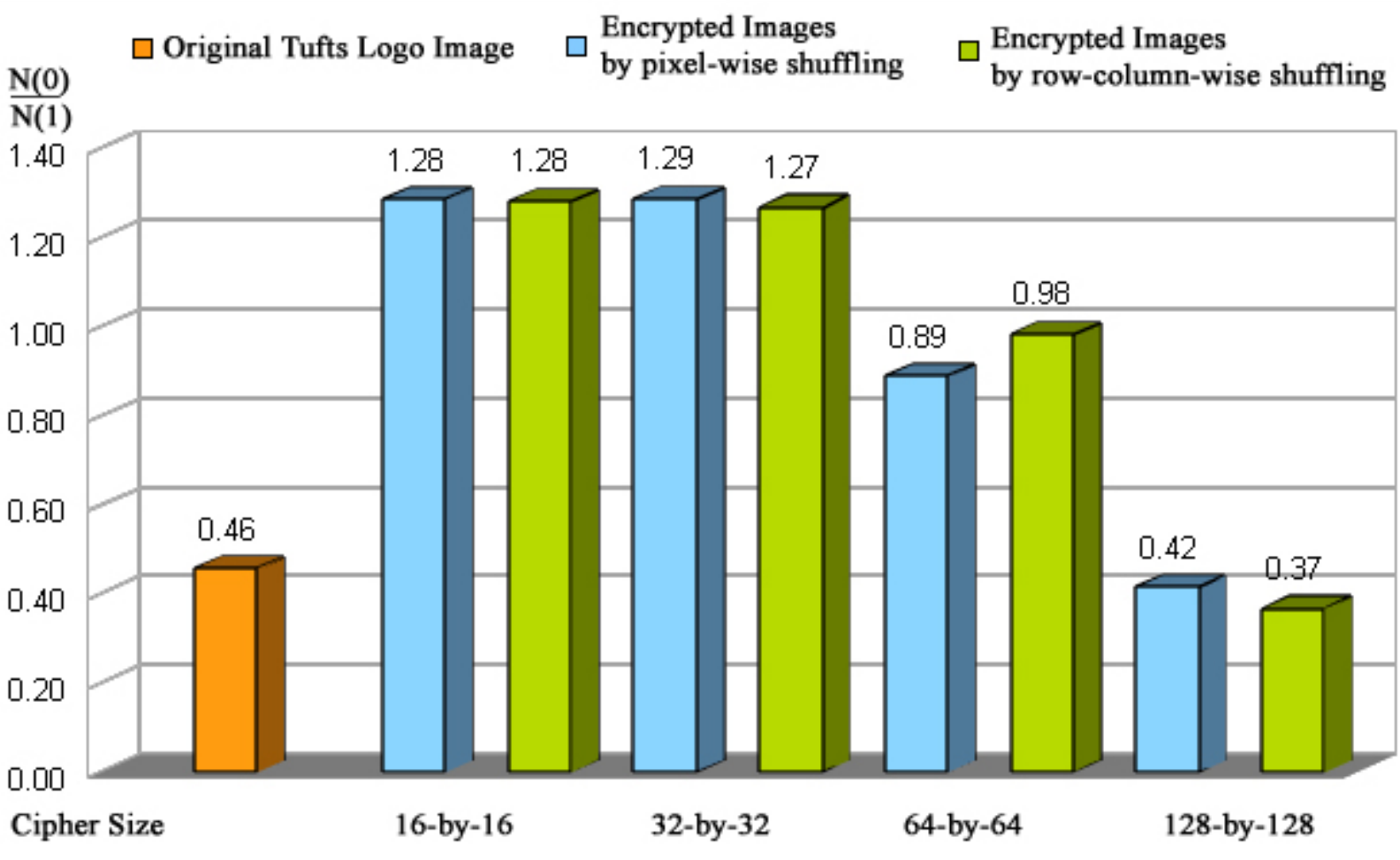}}
     \centerline{\scriptsize (a) Tendency of $T=\#0bits/\#1bits$}
      \vspace{0.12cm}
    \end{minipage}\hfill
    \begin{minipage}[b]{0.48\linewidth}
      \centering
     \centerline{\includegraphics[width=\linewidth]{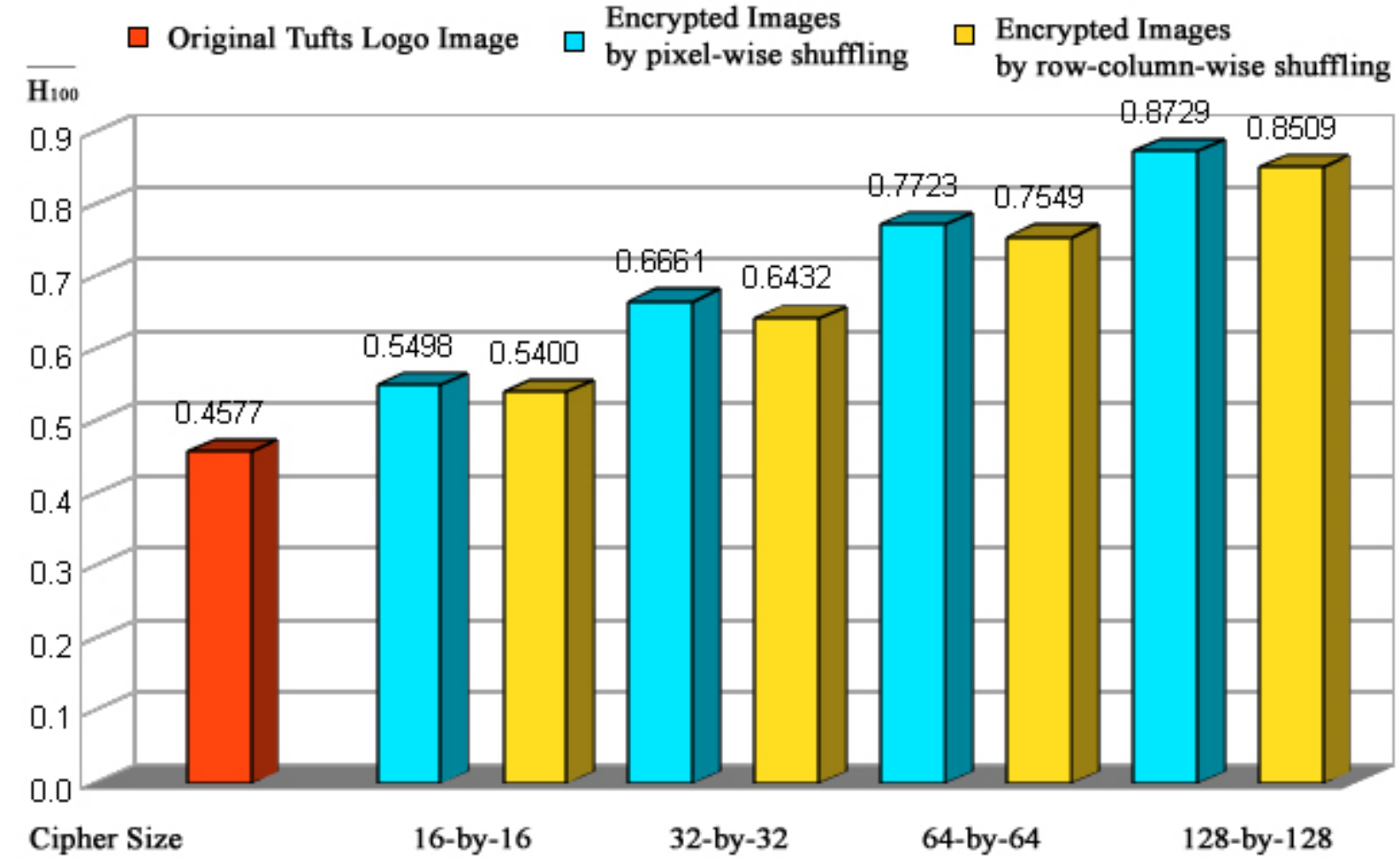}}
     \centerline{\scriptsize (b) Tendency of $\overline {H_{100}}$}
      \vspace{0.12cm}
    \end{minipage}\hfill
    \caption{Randomness Test for Image Shuffling}
\end{figure}
%

Although the binary Tufts logo is not random and neither are any of its shuffled versions, the objective of applying these two set randomness tests is to investigate whether some test randomness score(s) could be meaningful in the sense of coinciding with the randomness tendencies observed by using human visual inspections. The results of FIPS 140-2 tests and the proposed block entropy test are shown in Table \ref{Randomness Table}. It is worthwhile to note the required binary sequence length for FIPS 140-2 tests is 20,000. In our tests, a rectangular region of interest (ROI) at size of 100-by-200 is randomly selected for each test image, then the test binary sequence is extracted one by one from this ROI and is tested with FIPS 140-2 randomness tests. Meanwhile, the size $M$-by-$N$ for the block entropy test is set to 16-by-16 and the sample size $K$ is 100.

%

From the column of 'Block Entropy Test' in Table \ref{Randomness Table}, it is clear that the block entropy test scores matches our expectation of image randomness from the point view of human visual inspection. Although one statistic $T$, the ratio of the number of 0 bits to the number of 1 bits, roughly supports the two conclusions, this ratio is indirect and implicit in the randomness test of Monobit. Unless the ratio $T$ of the original binary 'Tufts' logo is known as 0.41, it is hard to draw conclusions  simply by looking at Fig. 5-(a). However, the two conclusions are pretty obvious with results of the block entropy test as Fig. 5-(b) shows.


From this example, it is clear that the proposed block entropy test is applicable to image shuffling problems as well and the test score indicates the averaged randomness of the shuffled image. A higher block entropy test score implies a higher randomness for local image blocks and thus a better shuffling quality. Compared to conventional randomness test tools for binary sequences, the block entropy test is more robust for testing image randomness. Meanwhile, the test scores do match results from human visual inspections.

\subsection{Test the Randomness of Image Encryption}
In this example, the block entropy test is applied to images encrypted by five image ciphers, including Blowfish \cite{Blowfish}, AES \cite{AES}, Twofish \cite{Twofish}, 3DCat \cite{3DCat} and Sudoku \cite{Sudoku} ciphers. Blowfish, AES and Twofish are well-known block ciphers designed for binary sequence encryption and are still the prevailing encryption methods even if they are not designed for digital images. Commercial ciphers \cite{I-Cipher,pictureEncryption} implemented Blowfish, AES and Twofish algorithms are used in this simulation. '3DCat' refers to the image encryption method based on chaotic sequences generated from 3D chaotic cat maps. 'Sudoku' refers to the Sudoku image cipher, which employs the Sudoku matrix in conventional substitution and transposition ciphers. Images encrypted by using '3DCat' and 'Sudoku' are realized by using the original codes. The USC-SIPI 'Miscellaneous' image set is selected as our plaintext image set. Computer simulations are done under Windows XP environment with MATLAB R2010a.
\begin{figure*}[htb]

    \label{fig-SIPI}
      \centering
     \centerline{\includegraphics[width=12cm]{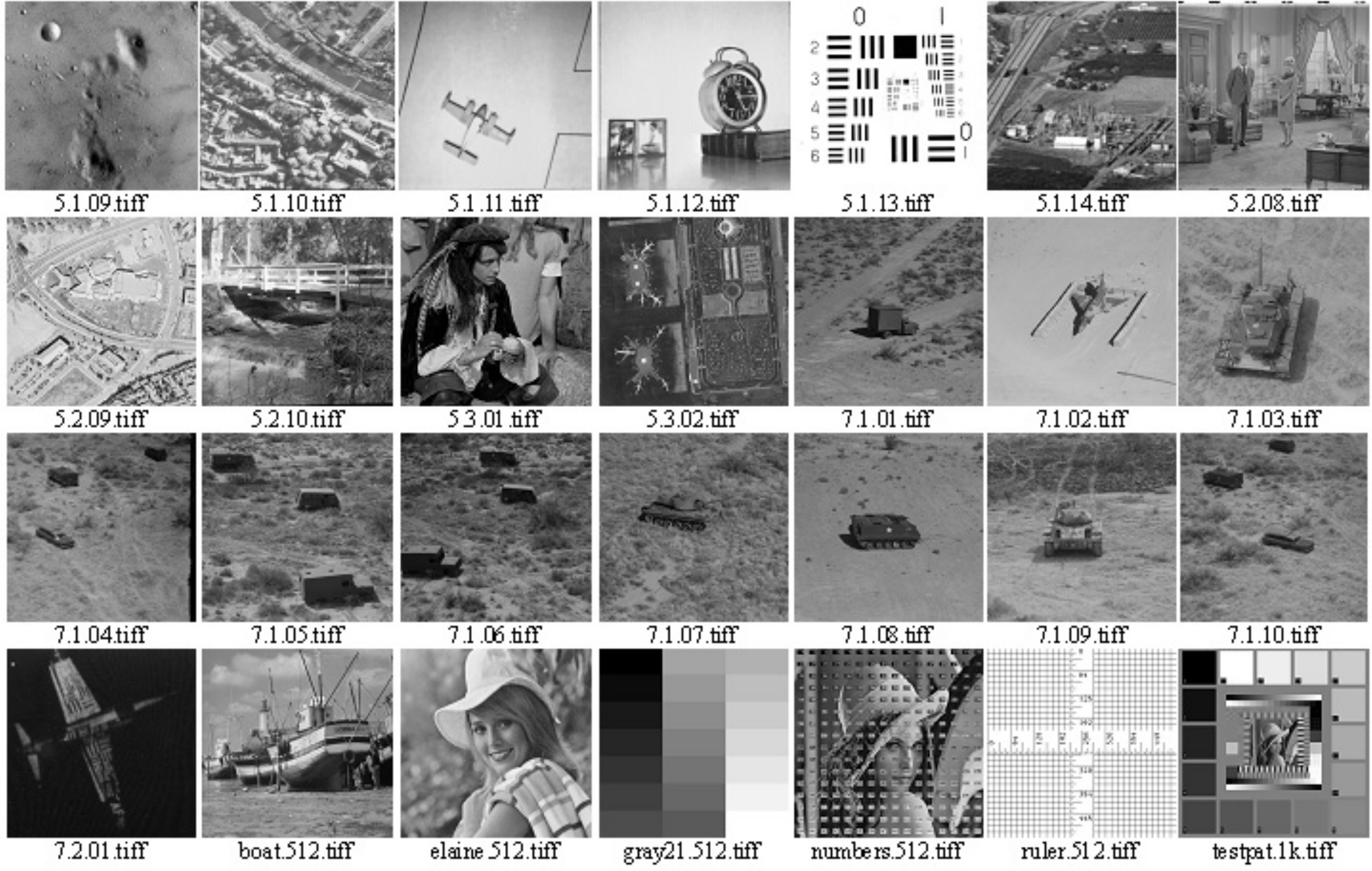}}
    \caption{Test Images from USC-SIPI 'Miscellaneous' Image Dataset}
\end{figure*}

Fig. 6 shows the total 28 test plaintext images covering a wide collection of image contents, scenes and patterns. These images are all gray with a size from 256-by-256 to 1024-by-1024. It is obvious that some man-made images in the plaintext image set are very challenging for encryption. For example, 'ruler.512.tif' is a gray image, but it is more similar to a binary image containing only black and white. Meanwhile, it has mesh-like patterns and large homogenous regions. All these properties imply that this image has a very tilted histogram with high correlations between rows, columns and image blocks. As a result, this image is considered as a hard case for image encryption.

In the following test, the 28 plaintext images first are encrypted by mentioned ciphers and then the block entropy test is applied to each encrypted image. It is worthwhile to note that related parameters are:
\begin{itemize}
    \item The block size used in test is 16-by-16, i.e. $MN = 256$
    \item The number of intensity scales is 256, i.e. $L = 256$
    \item The number of samples is 100, i.e. $K = 100$
\end{itemize}
Therefore, the theoretical distribution under ${\cal H}_0$, the critical values associated with $\alpha = 0.01$ and 100 samples in the hypothesis tests, and the corresponding upper bound of Type I error are:
\begin{itemize}
    \item Theoretical distribution under ${\cal H}_0$: $H_{100} \sim {\cal N}(\mu_{H_{100}^*}, \sigma_{H_{100}^*}^2)$, where $\mu_{H_{100}^*} = 7.1749663525$ and $\sigma_{H_{100}^*} = 0.00524379986$ (See Table 1)
    \item Critical value: $H_c = 7.1627674499$ with $\alpha = 0.01$ and 100 samples (See Table \ref{tabCriticalI}), if $\overline{H_{100}}<H_c$ then we say the test image is NOT ideally encrypted.
    \item Upper bound of Type I error when 100 is not considered as sufficiently large: $\gamma = 0.08654$ (See Table \ref{tabError})
\end{itemize}

Table 5 shows the block entropy test for images including original test images in Fig. 6 and encrypted images by mentioned ciphers. Data visualization for block entropy test results are provided in Fig. 7 (regions above $H_c$ is shaded). From this Table, it is easy to see that images encrypted by block ciphers for binary data tend to have a much lower mean $\overline{H_{100}}$ than those encrypted by image ciphers. On the other hand, $\overline{H_{100}}$ for image ciphers are around $\mu_{H_{100}^*}$ with much smaller standard deviations than those of block ciphers for binary data. This result support the claim that block ciphers for binary data are inappropriate to encrypt image data \cite{DateEncryption}.

Moreover, the $\overline{H_{100}}$ s that below  $H_c = 7.1627674499$ are marked with '*' in Table 5. The number of times of rejecting ${\cal H}_0$ for tested ciphers, denoted as $\#R$, are listed in the table. According the total number of images $T = 28$ with $\alpha = 0.01$ and $\gamma = 0.08654$, we know that:
\begin{itemize}
    \item If the tested cipher is ideal and $K = 100$ is sufficiently large, then the expected $\#R$ for a 28-image test equals $T\alpha = 0.28$. In other words, for any cipher which has a number of rejections bigger than 1, it is NOT ideal.
    \item If the tested cipher is ideal and $K = 100$ is NOT sufficiently large, then the upper bound of the expected $\#R$ for a 28-image equals $T\gamma = 2.42$. In other words, for any cipher which has a number of rejections bigger than 3, it is NOT ideal.
\end{itemize}
Therefore, based on this knowledge, all tested block ciphers for binary data do not generate ideally encrypted images for image data. '3DCat' cipher is a borderline because if the sample size is considered to be sufficiently large, then it is not of a type of ideal cipher which always generates ideally encrypted images, otherwise it is. 'Sudoku' cipher always give random-like ciphertext in this 28-image test and thus no evidence to reject that 'Sudoku' cipher is an ideal cipher for image data.

\begin{table}[htbp]
\begin{center}
\hfill{}
\label{Entropy Table}
\caption{The Block Entropy Test for Encrypted Images using 100 image blocks at size of 16-by-16}
\scriptsize
\begin{tabular}{|c||c||c|c|c||c|c|}
\hline
\textbf{} & \textbf{} & \multicolumn{ 5}{c|}{ \textbf{\scriptsize Encryption Method}} \\ \hline
\scriptsize  \textbf{\scriptsize Image} &  \textbf{\scriptsize \textit{Original}} &  \textbf{\textit{\scriptsize BlowFish}} &  \textbf{\scriptsize \textit{AES}} & \scriptsize  \textbf{\textit{\scriptsize TwoFish}} & \scriptsize  \textbf{\textit{\scriptsize 3DCat}} & \scriptsize  \textbf{\textit{\scriptsize Sudoku}} \\ \hline\hline
\scriptsize  \textbf{\textit{\scriptsize 5.1.09}} & \scriptsize  5.896705* & \scriptsize  7.036605* & \scriptsize  7.017258* & \scriptsize  7.028879* & \scriptsize  7.171908 & \scriptsize  7.180215 \\ \hline
\scriptsize  \textbf{\textit{\scriptsize 5.1.10}} & \scriptsize  6.692363* & \scriptsize  7.033179* & \scriptsize  7.039004* & \scriptsize  7.035334* & \scriptsize  7.179439 & \scriptsize  7.169024 \\ \hline
\scriptsize  \textbf{\textit{\scriptsize 5.1.11}} & \scriptsize  5.544970* & \scriptsize  7.020790* & \scriptsize  7.038881* & \scriptsize  7.033139* & \scriptsize  7.182150 & \scriptsize  7.177942 \\ \hline
\scriptsize  \textbf{\textit{\scriptsize 5.1.12}} & \scriptsize  5.557432* & \scriptsize  7.030360* & \scriptsize  7.019917* & \scriptsize  7.038484* & \scriptsize  7.173553 & \scriptsize  7.175398 \\ \hline
\scriptsize  \textbf{\textit{\scriptsize 5.1.13}} & \scriptsize  1.158392* & \scriptsize  4.517699* & \scriptsize  4.942020* & \scriptsize  4.917323* & \scriptsize  7.176973 & \scriptsize  7.186245 \\ \hline
\scriptsize  \textbf{\textit{\scriptsize 5.1.14}} & \scriptsize  6.681083* & \scriptsize  7.039437* & \scriptsize  7.043030* & \scriptsize  7.043213* & \scriptsize  7.175974 & \scriptsize  7.165637 \\ \hline
\scriptsize  \textbf{\textit{\scriptsize 5.2.08}} & \scriptsize
 6.017607* & \scriptsize  7.171057 & \scriptsize  7.168529 & \scriptsize  7.169126 & \scriptsize  7.183370 & \scriptsize  7.178135 \\ \hline
\scriptsize  \textbf{\textit{\scriptsize 5.2.09}} & \scriptsize  6.228063* & \scriptsize  7.172360 & \scriptsize  7.168744 & \scriptsize  7.173990 & \scriptsize  7.165557 & \scriptsize  7.173511 \\ \hline
\scriptsize  \textbf{\textit{\scriptsize 5.2.10}} & \scriptsize  5.045033* & \scriptsize  7.166008 & \scriptsize  7.175324 & \scriptsize  7.164280 & \scriptsize  7.168637 & \scriptsize  7.173879 \\ \hline
\scriptsize  \textbf{\textit{\scriptsize 5.3.01}} & \scriptsize  5.923279* & \scriptsize  7.169365 & \scriptsize  7.173941 & \scriptsize  7.173221 & \scriptsize  7.176639 & \scriptsize  7.170577 \\ \hline
\scriptsize  \textbf{\textit{\scriptsize 5.3.02}} & \scriptsize  5.828853* & \scriptsize  7.174636 & \scriptsize  7.174182 & \scriptsize  7.170160 & \scriptsize  7.172147 & \scriptsize  7.177803 \\ \hline
\scriptsize  \textbf{\textit{\scriptsize 7.1.01}} & \scriptsize  5.256613* & \scriptsize  7.177397 & \scriptsize  7.168727 & \scriptsize  7.168608 & \scriptsize  7.146194* & \scriptsize  7.176737 \\ \hline
\scriptsize  \textbf{\textit{\scriptsize 7.1.02}} & \scriptsize  3.130283* & \scriptsize  7.163421 & \scriptsize  7.167083 & \scriptsize  7.163486 & \scriptsize  7.182047 & \scriptsize  7.182979 \\ \hline
\scriptsize  \textbf{\textit{\scriptsize 7.1.03}} & \scriptsize  5.015711* & \scriptsize  7.163324 & \scriptsize  7.175095 & \scriptsize  7.172171 & \scriptsize  7.170552 & \scriptsize  7.173915 \\ \hline
\scriptsize  \textbf{\textit{\scriptsize 7.1.04}} & \scriptsize  5.233899* & \scriptsize  7.164166 & \scriptsize  7.166298 & \scriptsize  7.166104 & \scriptsize  7.141284* & \scriptsize  7.181730 \\ \hline
\scriptsize  \textbf{\textit{\scriptsize 7.1.05}} & \scriptsize  5.990831* & \scriptsize
 7.173422 & \scriptsize  7.168685 & \scriptsize  7.113663* & \scriptsize  7.169440 & \scriptsize  7.176741 \\ \hline
\scriptsize  \textbf{\textit{\scriptsize 7.1.06}} & \scriptsize  6.101809* & \scriptsize  7.166029 & \scriptsize  7.164576 & \scriptsize  7.167572 & \scriptsize  7.171974 & \scriptsize  7.168835 \\ \hline
\scriptsize  \textbf{\textit{\scriptsize 7.1.07}} & \scriptsize  5.534905* & \scriptsize  7.161739* & \scriptsize  7.171792 & \scriptsize  7.167448 & \scriptsize  7.165520 & \scriptsize  7.171042 \\ \hline
\scriptsize  \textbf{\textit{\scriptsize 7.1.08}} & \scriptsize  4.268606* & \scriptsize  7.167514 & \scriptsize  7.172741 & \scriptsize  7.167774 & \scriptsize  7.178153 & \scriptsize  7.182649 \\ \hline
\scriptsize  \textbf{\textit{\scriptsize 7.1.09}} & \scriptsize  5.543805* & \scriptsize  7.164298 & \scriptsize  7.164665 & \scriptsize  7.164676 & \scriptsize   7.181045 & \scriptsize  7.177071 \\ \hline
\scriptsize  \textbf{\textit{\scriptsize 7.1.10}} & \scriptsize  5.366133* & \scriptsize  7.164842 & \scriptsize 7.172609 & \scriptsize  7.166904 & \scriptsize  7.175036 & \scriptsize  7.174440 \\ \hline
\scriptsize  \textbf{\textit{\scriptsize 7.2.01}} & \scriptsize  4.724605* & \scriptsize  7.183806 & \scriptsize  7.169110 & \scriptsize  7.167013 & \scriptsize  7.172266 & \scriptsize   7.172823 \\ \hline
\scriptsize  \textbf{\textit{\scriptsize boat.512}} & \scriptsize  6.165926* & \scriptsize  7.159163* & \scriptsize  7.174200 & \scriptsize  7.166788 & \scriptsize  7.176104 & \scriptsize  7.174827 \\ \hline
\scriptsize  \textbf{\textit{\scriptsize elaine.512}} & \scriptsize  6.246007* & \scriptsize  7.166420 & \scriptsize  7.161446* & \scriptsize  7.161411* & \scriptsize  7.168059 & \scriptsize  7.174350 \\ \hline
\scriptsize  \textbf{\textit{\scriptsize gray21.512}} & \scriptsize  1.949937* & \scriptsize  5.208080* & \scriptsize  5.392334* & \scriptsize  5.455673* & \scriptsize  7.165595 & \scriptsize  7.199380 \\ \hline
\scriptsize  \textbf{\textit{\scriptsize numbers.512}} & \scriptsize  5.898473* & \scriptsize  7.159820* & \scriptsize  7.168180 & \scriptsize  7.173836 & \scriptsize  7.173024 & \scriptsize  7.167303 \\ \hline
\scriptsize  \textbf{\textit{\scriptsize ruler.512}} & \scriptsize  0.330469* & \scriptsize  2.132499* & \scriptsize  2.468986* & \scriptsize  2.562190* & \scriptsize  7.175261 & \scriptsize  7.187348 \\ \hline
\scriptsize  \textbf{\textit{\scriptsize testpat.1k}} & \scriptsize  1.625584* & \scriptsize  4.323696* & \scriptsize  4.779638* & \scriptsize  4.781003* & \scriptsize  7.173987 & \scriptsize  7.192987 \\ \hline\hline
\scriptsize  \textbf{Statistics} & \scriptsize  \scriptsize  \textbf{\textit{Original}} & \scriptsize  \scriptsize  \textbf{\textit{\scriptsize BlowFish}} & \scriptsize
 \scriptsize  \textbf{\textit{\scriptsize AES}} & \scriptsize  \scriptsize  \textbf{\textit{TwoFish}} &
  \textbf{\textit{\scriptsize 3DCat}} &  \textbf{\scriptsize \textit{Sudoku}} \\ \hline
\scriptsize  \textbf{\textit{\scriptsize Average}} & \scriptsize  4.9627633737 & \scriptsize  6.6975405565 & \scriptsize  6.7488211774 & \scriptsize  6.7511953801 &\scriptsize
 7.1718531782 & \scriptsize  7.177268703 \\ \hline
\scriptsize  \textbf{\textit{\scriptsize Std}} & \scriptsize  1.7133301636 & \scriptsize  1.1895800576 & \scriptsize  1.0726741341 & \scriptsize  1.0564646259 & \scriptsize  0.0093951696 & \scriptsize  0.0075774206 \\ \hline
\scriptsize  \textbf{\textit{\scriptsize \# Rejection ${\cal H}_0$}} & \scriptsize  28 & \scriptsize  12 & \scriptsize  10 & \scriptsize  11 & \scriptsize  2 & \scriptsize  0 \\ \hline
\end{tabular}
\hfill{}
\end{center}
\end{table}

\begin{figure}[htb]
    \label{fig-SIPI-results}
      \centering
     \centerline{\includegraphics[width=15cm]{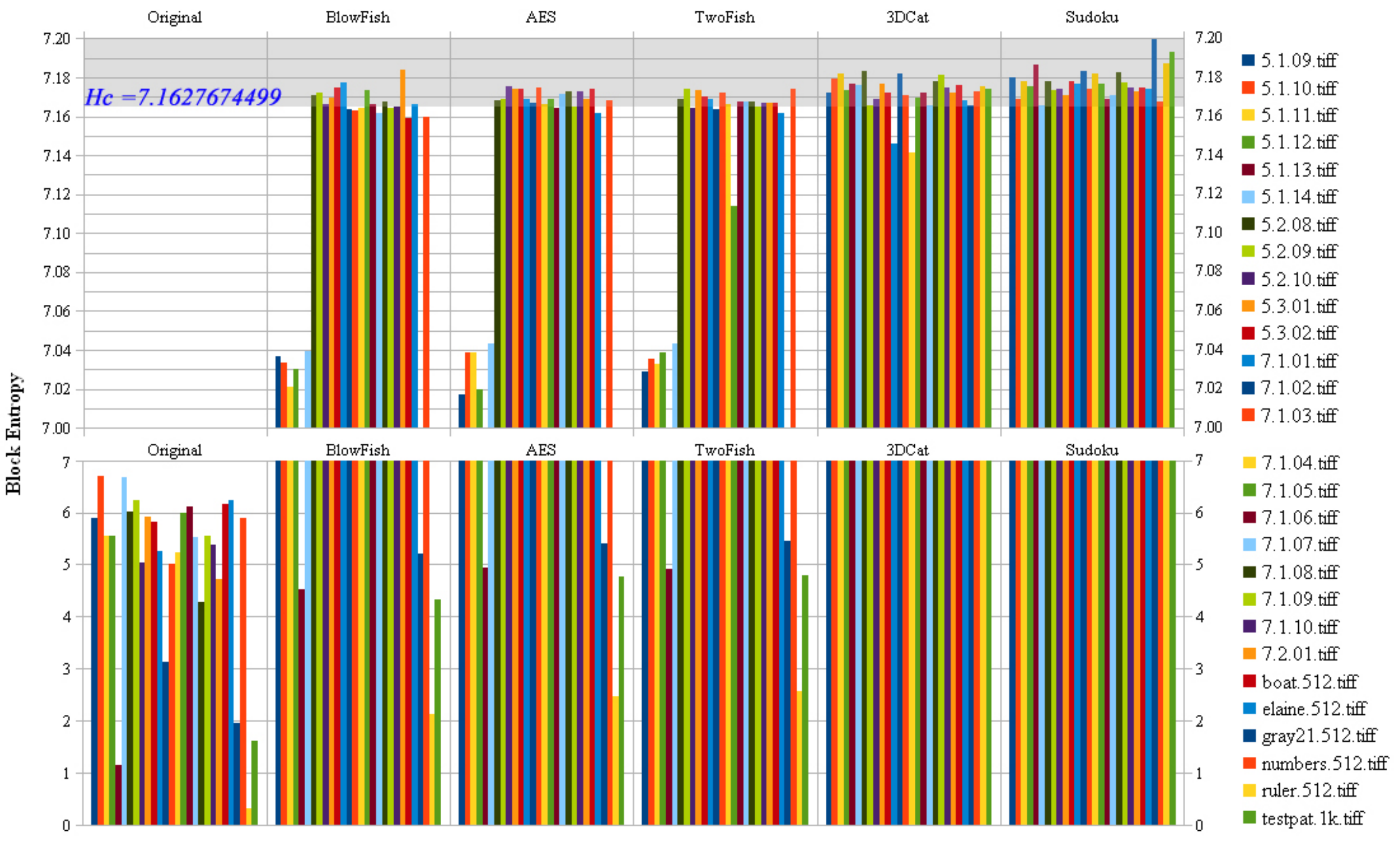}}
    \caption{The Block Entropy Test for Encrypted Images using 100 gray image blocks at size 16-by-16 }
\end{figure}
\section{Conclusion}
In this paper, the block entropy test for image encryption is proposed. Unlike the conventional global entropy test for image encryption, it measures the sample mean $\overline{H_K}$ of information entropies for $K$ non-overlapped image blocks within a test image, then compare this $H_K$ with reference theoretical values to conclude whether the test image is ideally encrypted. Therefore, it has a better ability to measure the randomness of local image blocks within a test image and gives more reliable randomness measures than the global entropy test. In this way, the block entropy test contains all four desired properties for testing image randomness.

The theoretical mean $\mu_{H_K^*}$ and variance $\sigma_{H_K^*}^2$ of $\overline{H_K}$ are derived from the ideally encrypted image/true random image model. This assumption states that the random variable of any pixel's intensity follows a discrete uniform distribution on $[0,L-1]$ for an ideally encrypted image, where $L$ is the number of allowed intensity scales related to the image format.
As a result, the quantitative test result can be obtained by comparing the actual test $\overline{H_K}$ with theoretical mean $\mu_{H_K^*}$

Moreover, an $\alpha$-level of significance test is also derived. As a result, the qualitative test result becomes to check whether or not the test $\overline{H_K}$ is below the critical value $H_c$ associated with the $\alpha$-level of significance. The upper bound of Type I error $\gamma$ for the proposed hypothesis test is given by using the BET, for the case when the sample size $K$ in test is insufficiently large,

In the application section, two examples are shown. The first example focuses on quantitatively measuring the randomness of scrambled images using the block entropy test. Results show that the block entropy scores for test images match randomness extents given by human visual inspections. In contrast, the conventional FIPS 140-2 randomness tests fail to give 'reasonable' randomness scores for the same set of images. This implies that the block entropy test is effective and meaningful in measuring image randomness quantitatively. The second example concentrates on qualitatively testing whether encrypted images are random-like. Gray images from the USC-SIPI image data set are used in simulation, because these images covers a wide range of contents, textures, types etc. Simulation results of the block entropy test for five image ciphers show that the block ciphers for binary data are inappropriate for image data because their quantitative $\overline{H_{100}}$ are a way lower than the theoretical mean $\mu_{H_{100}^*}$ for many test images. Furthermore, the qualitative results of block ciphers for binary data are very poor as well. In contrast, image ciphers '3Dcat' and 'Sudoku' performance much better. Especially noteworthy, the 'Sudoku' image cipher is the only test cipher that always generates random-like ciphertext images. These results show that the block entropy test is effective and robust.

\bibliographystyle{model1b-num-names}
\bibliography{Ref}








\end{document}